\documentclass[twocolumn]{IEEEtran}
\usepackage{graphicx,subfigure}
\usepackage{enumerate}
\usepackage{color}
\usepackage{pifont}
\usepackage{float}
\usepackage{epsfig}
\usepackage[makeroom]{cancel}
\usepackage{multirow}
\usepackage{pbox}
\usepackage{lipsum}
\usepackage{cleveref}
\usepackage{cite,amssymb,amsthm}
\usepackage{tablefootnote}
\usepackage{amsmath}
\usepackage{arydshln}
\usepackage[ruled, vlined]{algorithm2e}

\newcommand{\mup}{\mu \text{PMU}}
\newcommand{\bs}{\boldsymbol} 

\newcommand{\mb}{\mathbf}
\usepackage{color}
\newtheorem{prop}{Proposition}
\begin{document}

\title{Anomaly Detection Using Optimally-Placed $\mup$ Sensors in Distribution Grids}

\author{Mahdi Jamei, Anna Scaglione, Ciaran Roberts, Emma Stewart, \\ Sean Peisert, Chuck McParland, Alex McEachern   
    \thanks{This research was supported in part by the Director, Office of Electricity Delivery and Energy Reliability, Cybersecurity for Energy Delivery Systems program, of the U.S. Department of Energy, under contracts DE-AC02-05CH11231 and DEOE0000780. Any opinions, and findings expressed in this material are those of the authors and do not necessarily reflect those of the sponsors.). 
Preliminary version of this work has been accepted to be published in the proceedings of HICSS 2017, Hawaii, USA \cite{jamei2016automated}.}
        \thanks{M. Jamei and A. Scaglione are with the School of Electrical, Computer and Energy Engineering, Arizona State University, Tempe, AZ, USA. Emails: \texttt{\{mjamei, ascaglio\}@asu.edu}. C. Roberts, E. Stewart, S. Peisert, and C. McParland are with the Lawrence Berkeley National Laboratory, Berkeley, CA, USA. Emails: \texttt{\{cmroberts, estewart, sppeisert, cpmcparland\}@lbl.gov}. A. McEachern is the CEO of the Power Standards Lab, Alameda, CA, USA. Email: \texttt{Alex@McEachern.com}.}
         }
\maketitle

\begin{abstract}
As the distribution grid moves toward a tightly-monitored network, it is important to automate the analysis of the enormous amount of data produced by the sensors to increase the operators situational awareness about the system. 
In this paper, focusing on Micro-Phasor Measurement Unit ($\mup$) data, we propose a hierarchical architecture for monitoring the grid and establish a set of analytics and sensor fusion primitives for the detection of abnormal behavior in the control perimeter.    
Due to the key role of the $\mup$ devices in our architecture, a source-constrained optimal $\mup$ placement is also described that finds the best location of the devices with respect to our rules. The effectiveness of the proposed methods are tested through the synthetic and real $\mup$ data.
%
\end{abstract}
\IEEEpeerreviewmaketitle
\section{Introduction}
Historically, Distribution System Operators (DSOs) have lacked information about the real-time operation of the grid. However, this situation is changing. Similar to the Phasor Measurement Units (PMUs) at the transmission level, Micro-Phasor Measurement Units ($\mup$s) are designed to measure real-time voltage and current phasors in a synchronized manner at specific locations on the distribution system. These measurements can be utilized by the DSO, but useful translation from raw data to information is first required \cite{von2014micro, scoping_study}. This information can provide a level of insight that is not attainable by Distribution Supervisory Control and Data Acquisition (DSCADA) data. Visualizing and interpreting raw sensor streams can be overwhelming for DSOs, considering the large quantity of data that could flow in from different parts of the grid \cite{kezunovic2013role}. Therefore, it is essential to mine the data collected with analytic tools that can derive informative measurements and form automated reports.
\subsection{Literature Review} 
Recent work has focused on the PMU data utilization at the transmission level to improve the Wide-Area Monitoring, Protection, and Control \cite{phadke2008wide,terzija2011wide}. Distribution grids, however, are still lagging in that respect, since tools for the transmission grid may not be directly applicable to the distribution grid due to a different operation environment, such as load imbalances, untransposed lines, and the existence of single-phase and two-phase laterals. Transmission operations and system wide analysis are concerned with large imbalances in load and
generation, and as a result frequency, whereas distribution operations are concerned with localized, but
frequent events such as voltage imbalance, overloading, and outage management. The goal of this paper is to fill the distribution systems analysis gap by defining {\it physics-aware} algorithms that process $\mup$ data for the automatic detection of any grid behavior that does not comply with the (quasi) steady-state regime of operation, as well as the physical limits of the grid.       
Most of the prior research on sensor data analytics (including SCADA and PMU measurements) is concerned with detecting events on the grid transmission level. For example, Pan et al., \cite{pan2015developing} use data mining techniques on PMU measurements and audit logs for event classification. A linear basis expansion for the PMU data is described by Chen et al., \cite{chen2013dimensionality,xie2014dimensionality} for event detection application. A similar approach, based on Principal Component Analysis (PCA), is used in \cite{valenzuela2013real, ge2015power} for event detection and data archival.
Allen et al., \cite{allen2014pmu} describe the use of voltage phasor angle differences between different PMU readings to detect events. Biswal et al., \cite{biswalsupervisory} use the strongest signatures of event in PMU data for situational awareness enhancement. In our previous work \cite{jamei2016micro} we describe an architecture for a cyber-physical distribution grid, and show how event detection using $\mup$ measurements can help in reasoning about the security status of the grid.
A study by Brahma et al., \cite{brahmareal} describes the real-time dynamic event identification in power system using PMU data based on a data-driven and also physics-based method.
\subsection{Our Contribution}
Our approach for event detection is a combination of the data-driven methods, as well as criteria  resulting from analyzing the underlying physical model of the system.
We first define a hierarchical \textit{``anomaly detection architecture''} to host a set of anomaly detection rules that are proposed for the analysis and fusion of $\mup$ sensor measurements. Unlike model agnostic machine learning methods that look for statistical anomalies in a feature space that is often heuristic,
the \textit{anomalies} that algorithms in this paper identify are defined in the context of power quality and protection, in addition to what is imposed by the grid governing physical equations (i.e. Kirchhoff Voltage and Current Laws). This, in turn, gives a DSO much greater insight and help in the forensic analysis. Because of the important role played by $\mup$ data in the framework, a source-constrained $\mup$ placement optimization is described, to achieve the maximum sensitivity in detecting a change\footnote{Note that PMU placement methods studies for state estimation
are not applicable here since the state is not directly observable by solely depending on the $\mup$s due to the scarcity of sensors.}. 

\subsection{Notation}
The following notations are used throughout the paper:
\begin{tabular}{ll}
$j$&Imaginary unit.\\
$\mathcal{I}_N$& $N \times N$ identity matrix.\\
$\mb{A}^T, \mb{A}^*,$&Transpose, conjugate,\\
$\mb{A}^H$&and conjugate transpose of matrix $\mb{A}$.\\
$||\mb{A}||,||\mb{A}||_F$&2-norm and F-norm of matrix $\mb{A}$.\\ 
$\mb{1}_{m \times n}$&$m \times n$ size matrix with entries 1.\\
$\mb{A}^\dagger$&Pseudo-inverse of matrix $\mb{A}$.\\
$\otimes$&Kronecker product.  
\end{tabular}

\section{Anomaly Detection Architecture}
As shown in Fig.~\ref{fig.anomaly_arch} the first stage of the  \textit{``anomaly detection architecture''} hierarchy is next to each $\mup$ sensor in the field and the second stage is at a central node that can be hosted in the Distribution System Operator (DSO) control room. The algorithms for anomaly detection applied next to each $\mup$ sensor are referred to as  
\textit{``local rules''} and those that aggregate readings of multiple $\mup$s as \textit{``central rules''}.         
\begin{figure}[ht]
\centering
\includegraphics[trim = 0mm 0mm 0mm 0mm,width=0.8\linewidth]{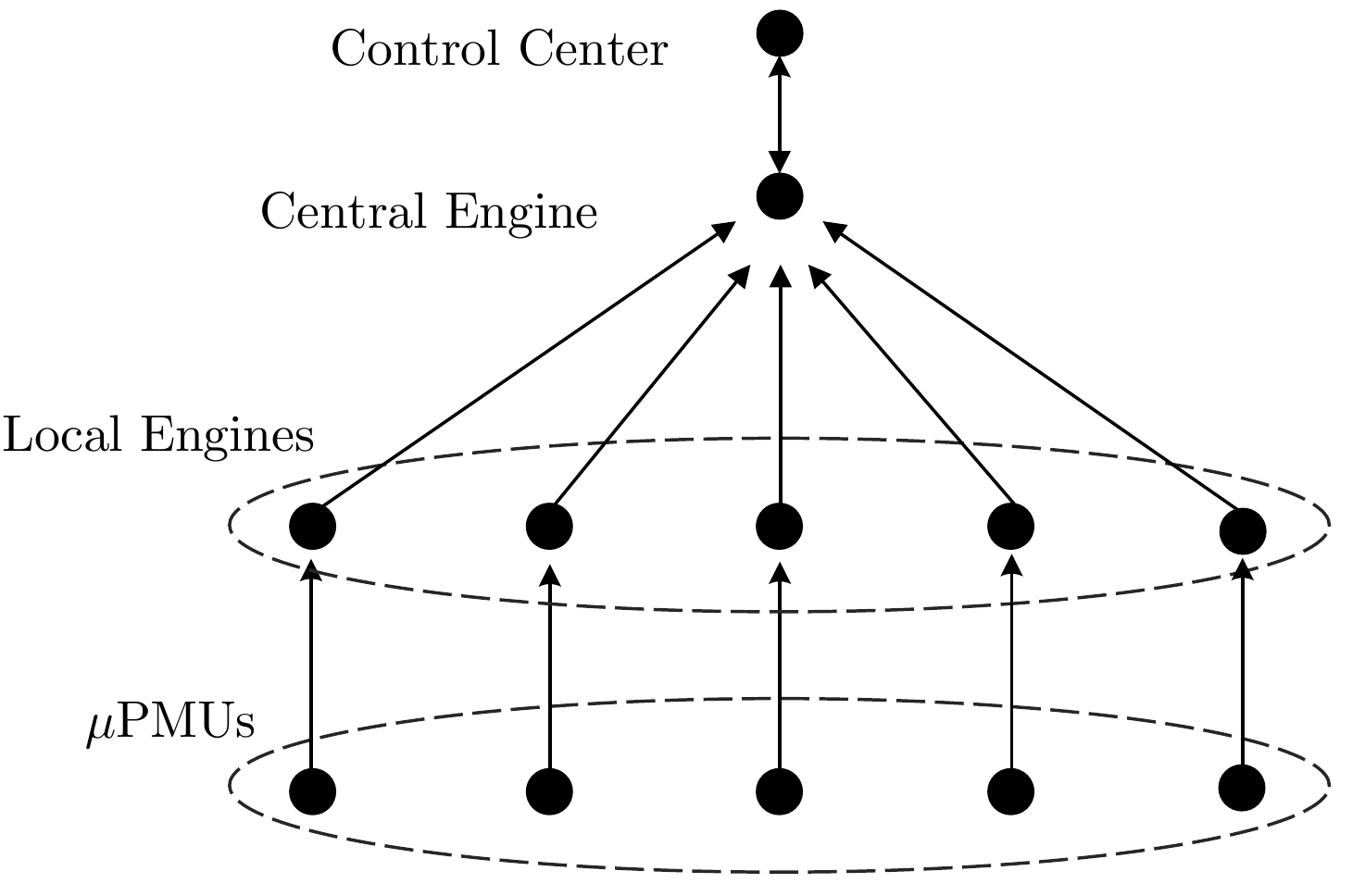}
\caption{Anomaly Detection Architecture Using $\mup$ Measurements.}
\label{fig.anomaly_arch}
\end{figure}
For large grid sizes, the aggregation can occur in multiple steps, where mid-level stages analyze part of the data and forward them upward.

This design has the following benefits: (1) it is robust due to its distributed, redundant nature; (2) it can be inexpensively deployed at existing utilities; (3) it analyzes the data close to real-time; as opposed to the common practicing of post hoc auditing of logged data; (4) the employed $\mup$s are placed in separate protection domains---including a separate communications network---from the DSCADA systems, so that attacks against the DSCADA network will not compromise the operation of the architecture. In addition, $\mup$s can be seamlessly added to the network via wireless modems. Further addressing security, the $\mup$s are not only designed to be read-only  via the network, but are both placed behind networked ``bastion hosts'' and can communicate bi-directionally via encrypted protocols, thereby fortifying the devices against tampering without physical access; (5) finally, even if a data injection attack can be launched against some $\mup$s, depending on level that is compromised (from field devices up to the control center), some $\mup$s and some of our rules may still be effective, due to the fact that every device in a certain level processes information independently. Thus, in our threat model, it is assumed that our detection rules are less-vulnerable to data injection attacks than attacks against typical DSCADA networks. 

The suite of algorithms that are proposed for the \textit{anomaly detection architecture} layers have the following advantages in comparison to the present state of the art: (1) due to the near-real-time analysis, analytic results can be used to prioritize the traffic flow from the lower to higher layer, pushing forward reports of anomalies faster than data that do not raise a flag and need to simply be accrued for historical purposes;
(2) it employs three-phase distribution grid equations rather than the more commonly-used positive sequence solution, thus avoiding the errors arising due to poor modeling; (3) a quasi steady-state condition is considered as the {\it normal regime} of operation rather than the idealistic steady-state which assumes there is no frequency drift. 
These modeling aspects are clarified in our discussion next.

\section{The $\mup$ Data in a Distribution Grid}\label{sec.model}       
Fig.~\ref{fig:dis_line} shows the $\pi$ model of a distribution line that connects bus $m$ to $n$.   
 \begin{figure}[ht]
 \centering
	\includegraphics[width=0.5\textwidth]{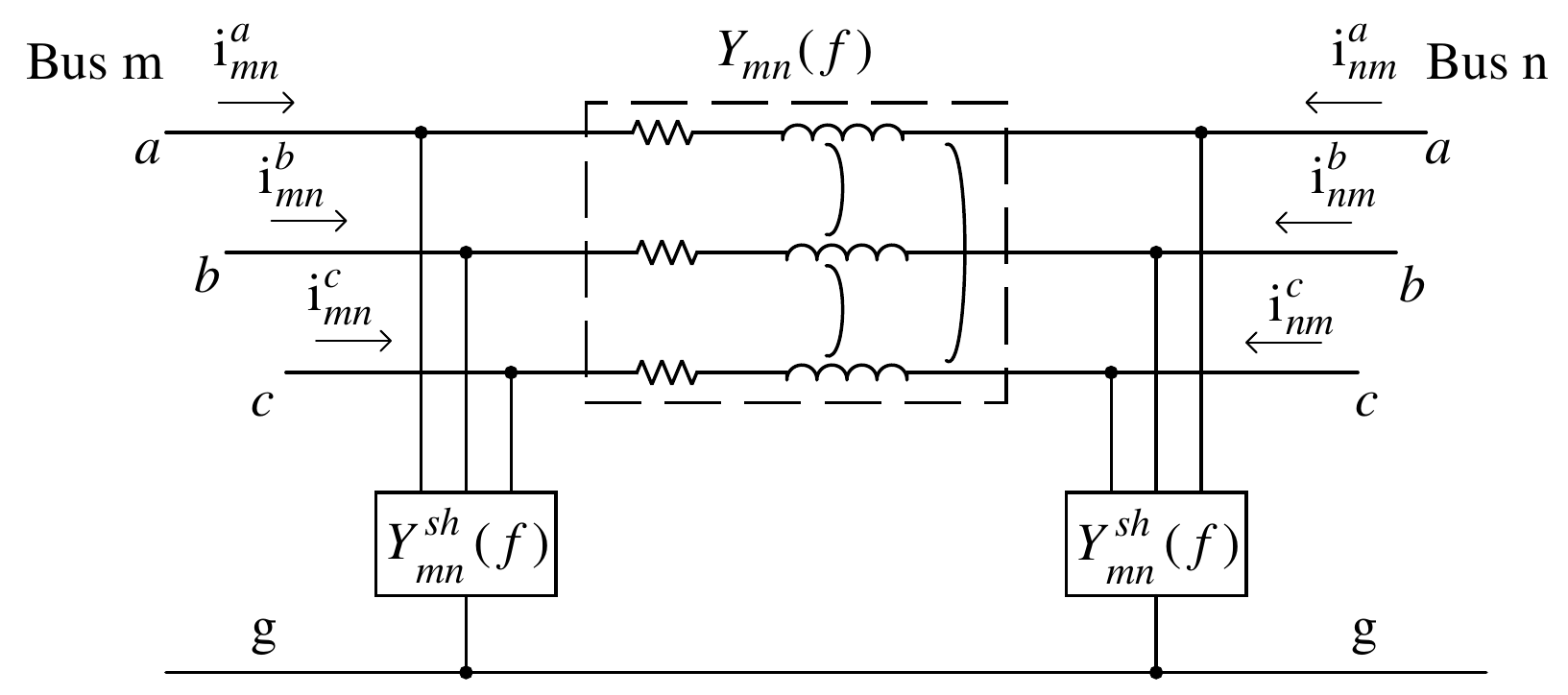}
	\caption{\small $\pi$ Model of a Distribution Line}
	\vspace{-0.2cm}
	\label{fig:dis_line}
\end{figure} 
Assuming normal conditions, the $\mup$s are designed to capture samples of the three phase voltage phasor, which is denoted by $\mb{v}[k]\in \mathbb{C}^{3\times 1}$, and of the current phasor $\mb{i}[k]\in \mathbb{C}^{3\times 1}$ in specific sites of a distribution grid. Next we apply Kirchhoff's and Ohm's law for a three-phase line in a quasi steady-state condition. 

A well-known fact from signals and systems theory is that the relationship between voltage and current through a passive circuit with a certain admittance matrix can be represented as a multiplication in the frequency domain and as a convolution in the time domain. Because the circuit is three phase, these will be represented by a Multi-Input Multi-Output system. This fact also holds for the phasor of the signals. We first define $\bs{y}_{mn}(t)$ and $\bs{y}^{sh}_{mn}(t)$ to denote the time domain equivalents of the matrices $\bs{Y}_{mn}(f+f_0)$ and $\bs{Y}^{sh}_{mn}(f+f_0)$, respectively.  
In discrete time, the convolution relationship is as follows:
\begin{align}
 \mathbf{i}_{mn}[k]=\sum_{r=0}^{N-1} \overline{\bs y}_{mn}[r] \mathbf{v}_m[k-r]-\bs{y}_{mn}[r]\mathbf{v}_n[k-r]
\label{eq:basiceq}
\end{align}
where $\overline{\bs y}_{mn}[r]=\bs{y}_{mn}[r]+\bs{y}^{sh}_{mn}[r]$ and it is assumed that $\bs{y}^{sh}_{ij}[n]$ and $\bs{y}_{mn}[r]$ are the samples respectively of  $\bs{y}^{sh}_{mn}(t)\star h(t)$ and $\bs{y}_{mn}(t)\star h(t)$, have finite support $N$ and are causal\footnote{$h(t)$ denotes the low-pass filter implemented in the $\mup$ to extract the baseband signal and $H(f)$ is its corresponding frequency response.
It should be noted that the outputs of  $\mup$, and their corresponding functions, are not the exact phasors if the bandwidth of the signal (voltage/current) is greater than $2f_0$, and are aliased. However, what we are interested in is to track any data abnormality (instead of the exact grid state during the event). Therefore,  as long as it is not suppressing the anomaly, having aliasing is not an issue for our rules.    
}. 

Next, the form of this relationship during the quasi steady-state is shown, since the steady-state in reality never happens, which in turn represents the governing equations during a regime of operation that is referred to as \textit{``normal''}. 

The fundamental frequency of the voltage and current signals are always varying, even in a normal state (although slowly and over a very small range), because of load-generation imbalances, active power demand interactions, large generators inertia, and the automatic speed controllers of the generators \cite{phadke2008synchronized}. The effect observable is a change of grid operation regime from steady-state to quasi steady-state. 
The off-nominal frequency therefore affects the phase angle captured by phasor measurement devices. To highlight that mathematically, the phasor readings are decomposed ${\mathbf{v}}_m[k]$ and $\mathbf{i}_{mn}[k]$ as follows:
\begin{align}
\mathbf{v}_m[k]=\hat{\mathbf{v}}_m[k] e^{j \beta_m[k]k}, 
~ \mathbf{i}_{mn}[k]=\hat{\mathbf{i}}_{mn}[k] e^{j \beta_m[k]k}
\label{eq:V_k_I_ij}
\end{align}
where $\hat{\mathbf{v}}_m[k]$ is the voltage phasor that is observed at nominal frequency, while $\hat{\mathbf{i}}_{mn}[k]$ is the current phasor after removing the exponential term due to the off-nominal frequency at bus $m$, $\beta_m[k]$, that represents the (time-varying) drift in the frequency induced by the above-mentioned reasons. 

Considering the Discrete Time Fourier Transform, we have:
\begin{align}
\bs{Y}_{mn}(f+f_0)H(f)=rect(T_sf)T_s\sum^{N-1}_{r=0}\bs{y}_{mn}[r]e^{-j2\pi rT_sf}
\end{align}
where $\frac{1}{T_s}=120$Hz is the $\mup$ output rate, and $H(f)$ is the frequency response of the filter $h(t)$. Introducing:
\begin{eqnarray}
{\bs{Y}}_{mn}(f_0,k)&\triangleq& \frac{1}{T_s} \bs{Y}_{mn}\!\!\left(f_0+\frac{\beta_n[k]}{2\pi T_s}\right)H\!\left(\frac{\beta_n[k]}{2\pi T_s}\right),
\label{eq:modulated_Y}
\end{eqnarray}
we have that 
\begin{equation}
{\bs{Y}}_{mn}(f_0,k) =\sum_{r=0}^{N-1}\bs{y}_{ij}[r]e^{-j\beta_n[k] r},
\end{equation}
and similarly $\overline{\bs Y}_{mn}(f_0,k)=\sum_{r=0}^{N-1}\overline{\bs{y}}_{mn}[r]e^{-j\beta_m[k] r}$.

During the quasi steady-state the variations in $\hat{\mathbf{v}}_m[k]$ and $\beta_m[k]$ are negligible over a window of $N$ samples, which means $\hat{\mathbf{v}}_m[k-r]\approx \hat{\mathbf{v}}_m[k]$ and $\beta_m[k-r]\approx \beta_m[k]$. Using this approximation, \eqref{eq:basiceq} can be re-written as:
\begin{align}
\begin{split}
\mathbf{i}_{mn}[k]&
\approx 
\sum_{r=0}^{N-1}\!\overline{\bs{y}}_{mn}[r]\hat{\mathbf{v}}_m[k]e^{j\beta_m[k](k-r)} \\&-\bs{y}_{mn}[r]\hat{\mathbf{v}}_n[k]e^{j\beta_n[k] (k-r)} \\
&=\overline{\bs Y}_{mn}(f_0,k)\mathbf{v}_m[k] -{\bs{Y}}_{mn}(f_0,k) \mathbf{v}_n[k]
\end{split}
\label{eq.i[k]}
\end{align}
%
Equation \eqref{eq.i[k]} is Ohm's law in the phasor domain and is the cornerstone for our transients detection algorithm derived in the following section. 
The analysis above explains through \eqref{eq.i[k]} why in the phasor domain, the equivalent effect of the quasi-steady state conditions is that the admittances \eqref{eq:modulated_Y} fluctuate. The effect is usually modest, because $\beta[k]$ is small. However, during a severe transient or frequency in the order of 10Hz the relationship \eqref{eq.i[k]} with the matrices in \eqref{eq:modulated_Y} does not hold, due to the manifestation of the full dynamic behavior in \eqref{eq:basiceq}.  

In the following, the \textit{``local''} and \textit{``central''} rules are defined leveraging these insights and the knowledge about power system operation.
Note that our rules are set up in a way that all the local engines are agnostic about $\overline{\bs Y}_{mn}(f_0,k)$
and $ \bs Y_{mn}(f_0,k)$, and the sensor siting. However, it is assumed that the central engine knows $\bs{Y}^{sh}_{mn}(f_0,k)|_{\beta[k]=0}$ and $\bs{Y}_{mn}(f_0,k)|_{\beta[k]=0}$ for the lines within its monitoring region and the difference from \eqref{eq:modulated_Y} is treated as equivalent to noise in the observation model. Therefore, when dealing with central rules, we will simply use $\bs{Y}^{sh}_{mn}$ and $\bs{Y}_{mn}$ to refer to
$\bs{Y}^{sh}_{mn}(f_0,k)|_{\beta[k]=0}$ and $\bs{Y}_{mn}(f_0,k)|_{\beta[k]=0}$. 

\section{Data Analysis}
\label{sec:data} 

Our rules monitor for abnormalities in the following quantities 1) voltage magnitude, 2) instantaneous frequency drift, 3) current magnitude, 4) active power, 5) reactive power, and 6) the validity of  quasi steady-state equations. 
The \textit{``local rules''} just require the stream of phasors from a single $\mup$, while the \textit{``central rules''} combine multiple streams across $\mup$s.


%

\subsection{Local Rules}\label{sec.stage-1} 

The \textit{``local rules''} are applied at the lowest layer of the \textit{``anomaly detection architecture''}, and run on systems on or immediately adjacent to the $\mup$s. Their common feature is that they 
require no specific prior knowledge of the grid network parameters. 
\begin{table}[h]
\centering
\caption{Voltage Magnitude Anomalies}
\begin{tabular}{|l|l|}
\hline
\text{anomaly}& \text{signature \tablefootnote{The voltage magnitude is in p.u.}}\\ \hline \hline
\text{voltage sag} & $0.1 \leq |\mathrm{v}| \leq 0.9,~{T_0}/{2} \leq \tau \leq 60 s$\\ \hline
\text{voltage swell} & $1.1 \leq |\mathrm{v}| \leq 1.8 ,~{T_0}/{2} \leq \tau \leq 60 s$\\ \hline
\text{interruption} & $ |\mathrm{v}| < 0.1,~{T_0}/{2} \leq \tau \leq 60 s$\\ \hline
\text{sustained interruption} & $|\mathrm{v}| < 0.1,~\tau > 60 s$\\ \hline
\text{undervoltage} & $0.1 \leq |\mathrm{v}| \leq 0.9,~\tau > 60 s$\\ \hline
\text{overvoltage} & $1.1 \leq |\mathrm{v}| \leq 1.8,~\tau > 60 s$\\ \hline
\end{tabular}
\label{table:vol_event}
\end{table}

\subsubsection{Voltage Magnitude Changes} The magnitude of the voltage varies within a small range that power quality standards enforce during the normal operations \cite{5154067}. Therefore, any large deviation from that range indicates an abnormal condition. Table.~\ref{table:vol_event} lists the anomalies that can be observed in the voltage magnitude labeled by their severity and duration, denoted by $|\mathrm{v}|$ and $\tau$, respectively.




\subsubsection{Current Magnitude, Active, and Reactive Power Changes} 
Suitable quantities to monitor for changes that can be computed locally from the $\mup$ voltage and current phasor measurements $\mathbf{v}_m[k]$ and $\mathbf{i}_{mn}[k]$ are current magnitude, active and reactive powers. 
The three phase apparent power can simply be computed as: 
\begin{align}
\mb{S}_{mn}[k]=\mb{P}_{mn}[k]+j\mb{Q}_{mn}[k]=\text{diag}(\mathbf{v}_m[k])\mathbf{i}^*_{mn}[k]
\label{eq:complex power}
\end{align}
where $\mb{P}_{mn}[k]$ and $\mb{Q}_{mn}[k]$ are the three-phase active and reactive powers, respectively. Note that for the sake of tracking a power flow change in the distribution grid, tracking the active power and reactive power is preferable over monitoring the phase angle difference, since the resistance of the lines is not negligible, and therefore the angle difference does not necessarily indicate the direction of the power flow. We observed this fact in our partner utility grid, when the voltage phase angle at one end of the line was less than the angle at the other end, though the direction of the power flow was not from the higher angle to the lower one.

Even when the voltage magnitude is within the safe range discussed previously, changes in active and reactive power can still happen due to the change of the load, affecting current magnitude and the phase angle between current and voltage phasors. Therefore, it is also of interest to track fast-changes of these quantities using the method described in Section.~\ref{sec:change_det}.

Because physical changes in this class of data on the distribution grid can be potentially indicative of negative behavior, it is important to determine the direction (upward trend, downward trend or oscillation) of the change.
%
The anomalies related to fast changes are labeled in this class of data with {\it  surge}, {\it drop}, and {\it oscillation} for increasing, decreasing, and swinging trends respectively by estimating the slope of the signal around the time of change.

While we have primarily introduced fast state changing events in both the dynamic and transient realms, we must also consider events in the steady state time frame, slower changing yet also potentially critical. An example of this could include a line rating or transformer load being slowly but consistently exceeded leading to accelerated failure. During the quasi steady-state, the three phase current phasor magnitude flowing in each line, i.e., $|\mb{i}_{mn}[k]|$, should be less than or equal the line rated current, $|\mb{i}_{mn}|_{\max}$. This constraint is imposed as feeder limit, and the violation is flagged as \textit{overcurrent}.


\subsubsection{Instantaneous Frequency Changes} For a $\mup$ at bus $m$, we propose to estimate adaptively the instantaneous local frequency deviation from the nominal frequency during the quasi-steady state, using the approach for instance in \cite{xia2012widely} that is tailored to three-phase distribution lines, to isolate abnormal changes in the estimated frequency. 
%
 
\subsubsection{Quasi Steady-State Regime Validity}
As previously noted, when the grid is not in the normal quasi steady-state conditions, the relationship between voltage and the current phasors represents its full dynamic behavior, i.e. the grid is no longer well-approximated by the set of memory-less algebraic equations. Therefore, it is proposed to check the validity of the quasi steady-state regime to flag the presence of transients in the grid. 
At the local engines, \eqref{eq.i[k]} provides the basis for our rule.
 
For the line in Fig.~\ref{fig:dis_line}, assuming that a $\mup$ is installed at bus $m$ means that $\mb{i}_{mn}[k]$, and $\mb{v}_m[k]$ are both available. 
Let $\bs{\alpha}[k]$ be the diagonal matrix such that voltage phasors of bus $m$ and bus $n$, connected via a power line, are related through: 
\begin{align}
\mb{v}_n[k]= \bs{\alpha}[k]\mb{v}_m[k],
\end{align}
Defining:
 \begin{align}
 \bs R^{(mn)}_{\rm iv}[k]&=\frac{1}{M-1}\sum_{r=0}^{M-1} \mb{i}_{mn}[k-r]\mb{v}^H_{m}[k-r],\\
  \bs R^{(nm)}_{\rm vv}[k]&=\frac{1}{M-1}\sum_{r=0}^{M-1} \mb{v}_{n}[k-r]\mb{v}^H_{m}[k-r],
 \end{align}
Assuming that $\bs{\alpha}[k]$ remains constant over a window of $M$ samples in the quasi-steady state, one can write:
 \begin{align}
\bs R^{(nm)}_{\rm vv}[k]\approx \bs{\alpha}[k] \bs R^{(mm)}_{\rm vv}[k]
 \end{align}
It can be assumed that the variation of $\overline{\bs Y}_{mn}(f_0,k)$ is negligible over $M$ samples in normal operation and use \eqref{eq.i[k]} to write:
\begin{align}
\nonumber
\! \! \left( \begin{array}{c:c}
\mathcal{I}_3 &
-\overline{\bs Y}_{mn}(f_0,k)+{\bs Y}_{mn}(f_0,k)\bs{\alpha}[k] 
\end{array}\right)
\underbrace{
 \begin{pmatrix}
  \bs R^{(mn)}_{\rm iv}[k]\\
    \bs R^{(mm)}_{\rm vv}[k]
 \end{pmatrix}
 }_{\bs R^{(mn)}_k}
\approx\bs 0
\label{eq.homog1}
\end{align}
 \begin{prop}
 \label{prop.Rk1}
Correlation matrix $\bs R^{(mn)}_k$ is approximately rank-1 during the quasi steady-state. 
 \end{prop}
 \begin{proof}
During the quasi-steady state along a distribution line, the following assumptions hold with a very good approximation for $r=0,1,...,M-1$:
\begin{align}
\hat{\mathbf{v}}_m[k-r] \approx \hat{\mathbf{v}}_m[k],~~~\beta_m[k-r] \approx \beta_m[k]
\end{align}  
Therefore, we can write:
\begin{align}
\begin{split}
\bs R^{(mm)}_{\rm vv}[k]&=\frac{1}{M-1}(\mathbf{v}_m[k] \otimes \mb{E}_m[k])(\mathbf{v}^H_m[k]\otimes \mb{E}_m^H[k])\\
&=\frac{1}{M-1}(\mathbf{v}_m[k]\mathbf{v}^H_m[k])\otimes(\mb{E}_m[k]\mb{E}_m^H[k])
\end{split}
\end{align}
where $\mb{E}_m[k]$ is defined as follows and represents the variations due to the off-nominal frequency:
\begin{align}
\mb{E}_m[k]=\mb{1}_{3\times 1}\otimes\begin{pmatrix}
e^{-j\beta_m[k](M-1)}&\ldots&e^{-j\beta_m[k]}&1\\
\end{pmatrix}
\end{align}
which we can then write:
\begin{align}
\mb{E}_m[k]\mb{E}_m^H[k]=(\mb{1}_{3\times 1}\mb{1}_{1\times 3})\otimes(M)=M\mb{1}_{3\times 3}
\end{align}
and therefore:
\begin{align}
\bs R^{(mm)}_{\rm vv}[k]=\frac{M}{M-1}(\mathbf{v}_m[k]\mathbf{v}^H_m[k])\otimes(\mb{1}_{3 \times 3})
\end{align}
which accordingly means that:
\begin{align*}
rank(\bs R^{(mm)}_{\rm vv}[k])=rank(\mathbf{v}_m[k]\mathbf{v}^H_m[k])\times rank(\mb{1}_{3 \times 3})=1 
\end{align*}

Because:
\begin{align*}
rank(\bs R^{(mn)}_k)=rank((\bs R^{(mn)}_k)^H\bs R^{(mn)}_k)
\end{align*}
we analyze the rank of $(\bs R^{(mn)}_k)^H\bs R^{(mn)}_k$ here, where:
\begin{align}
\begin{split}
(\bs R^{(mn)}_k)^H\bs R^{(mn)}_k=&(\bs R^{(mn)}_{\rm iv}[k])^H\bs R^{(mn)}_{\rm iv}[k]+\\
&(\bs R^{(mm)}_{\rm vv}[k])^H\bs R^{(mm)}_{\rm vv}[k]
\end{split}
\label{eq.RhR1}
\end{align}
From the structure of \eqref{eq.homog1} during the quasi-steady state, we have:
\begin{align}
\begin{split}
\bs R^{(mn)}_{\rm iv}[k]&=\tilde{\mb{Y}}_{mn}(f_0,k)\bs R^{(mm)}_{\rm vv}[k]\\
\tilde{\mb{Y}}_{mn}(f_0,k)&=\overline{\bs Y}_{mn}(f_0,k)-{\bs Y}_{mn}(f_0,k)\text{diag}(\bs{\alpha}[k])
\end{split}
\label{eq.Ytilde}
\end{align}
Substituting \eqref{eq.Ytilde} in \eqref{eq.RhR1}, we have:
\begin{align}
(\bs R^{(mn)}_k)^H\bs R^{(mn)}_k=(\bs R^{(mm)}_{\rm vv}[k])^H\bs{\mathcal{Y}}_{mn}(f_0,k)\bs R^{(mm)}_{\rm vv}[k]
\end{align}
where:
\begin{align*}
\bs{\mathcal{Y}}_{mn}(f_0,k)=\tilde{\mb{Y}}^H_{mn}(f_0,k)\tilde{\mb{Y}}_{mn}(f_0,k)+\mathcal{I} 
\end{align*}
Since the linear transformation of $\bs R^{(mm)}_{\rm vv}[k]$ does not increase its rank and since it has already been shown that $\bs R^{(mm)}_{\rm vv}[k]$ is of rank-1 during the quasi-steady state, one can conclude that:
\begin{align}
\begin{split}
&rank(\bs R^{(mn)}_k)=rank((\bs R^{(mn)}_k)^H\bs R^{(mn)}_k)\leq 
 \\
 & rank(\bs R^{(mm)}_{\rm vv}[k]) \rightarrow  rank(\bs R^{(mn)}_k)=1
\end{split}
\end{align}   
\end{proof}
Hence:
\begin{align}
\begin{split}
&\bs R^{(mn)}_k \approx \sigma^{(mn)}_1[k]\bs u^{(mn)}_1[k] (\bs \nu^{(mn)}_1[k])^H\rightarrow \\
&\bs R^{(mn)}_k (\bs R_k ^{(mn)})^H\approx (\sigma^{(mn)}_1[k])^2\bs u^{(mn)}_1[k] (\bs u^{(mn)}_1[k])^H
\end{split}
\end{align}
where $\sigma^{(mn)}_1$ is the largest singular value of $\bs R^{(mn)}_k$, and $\bs{u}^{(mn)}_1$ and $\bs{\nu}^{(mn)}_1$ are the corresponding left and right singular vectors to that, respectively. Deviation from this subspace structure can indicate that the line is experiencing a transient. 
We can automate the detection of anomaly using this criterion by computing the following cost minimization and tracking the fast changes in $x[k]$ for each incidental line to a bus with $\mup$:
\begin{equation}
\begin{split}
x[k]=&\min_{\bs u} ||(\mathcal{I}_{6}-\bs u\bs u^H)\bs R^{(mn)}_k(\bs R^{(mn)}_k)^H||_{F}\\
&\mbox{s.t.}~||\bs u||=1
\end{split}
\label{eq:cost_one_upmu}
\end{equation}
In other words, $x[k]$ measures the size of the residual that $\bs R^{(mn)}_k$ has in the space orthogonal to the optimal $\bs u$, which should be zero in the quasi steady-state and non-zero otherwise. Since it has already been shown that $\bs R^{(mn)}_k$ is of rank-1 during the quasi steady-state, the left singular vector corresponding to the largest singular value of $\bs R^{(mn)}_k$ is the only  quantity of interest to compute the metric, instead of computing all the singular vectors. Therefore, also owing it to the small size of the matrix, the method is not computationally very expensive.

To conclude the local rules, a flowchart of the analysis performed at the local engine next to $\mup$ at bus $m$ is presented in Fig.~\ref{fig.local_flowchart}. At each instant of time, the phasor readings are received by the local engine and the introduced metrics above are calculated. The pre-processed data are then passed to the local rules to check for any violation. A violation could be trespassing pre-defined limits (e.g., the voltage magnitude rule, or maximum current magnitude limit) or fast changes in a metric with smooth behavior during the normal condition (e.g., quasi steady-state validity rule). Once a violation is found in one of the metrics, the start time is recorded. The search for anomaly continues until no new violation is found for a certain window of time ({\it{``Count1 $>T_1$''}}), and that specifies the end time of the anomaly. The type of anomaly is then determined based on the behavior of the data between the start time and the end time (e.g., active power surge, voltage interruption,...). The start time, end time and the anomaly label is then sent upstream for further analysis/visualization, and the parameters are reset for next event. If the number of detected violations related to a certain event passes a pre-defined threshold ({\it{``Count2 $>T_2$''}}), the end time is replaced with a \textit{``Persistent''} label, and the results are sent to the central engine, without waiting for the end of the event to arrive. The reason is to be able to inform the operator about the anomaly in time, and not waiting too long before something more damaging happens.               
\begin{figure}
\centering
\includegraphics[trim = 1mm 1mm 1mm 1mm,width=0.8\linewidth]{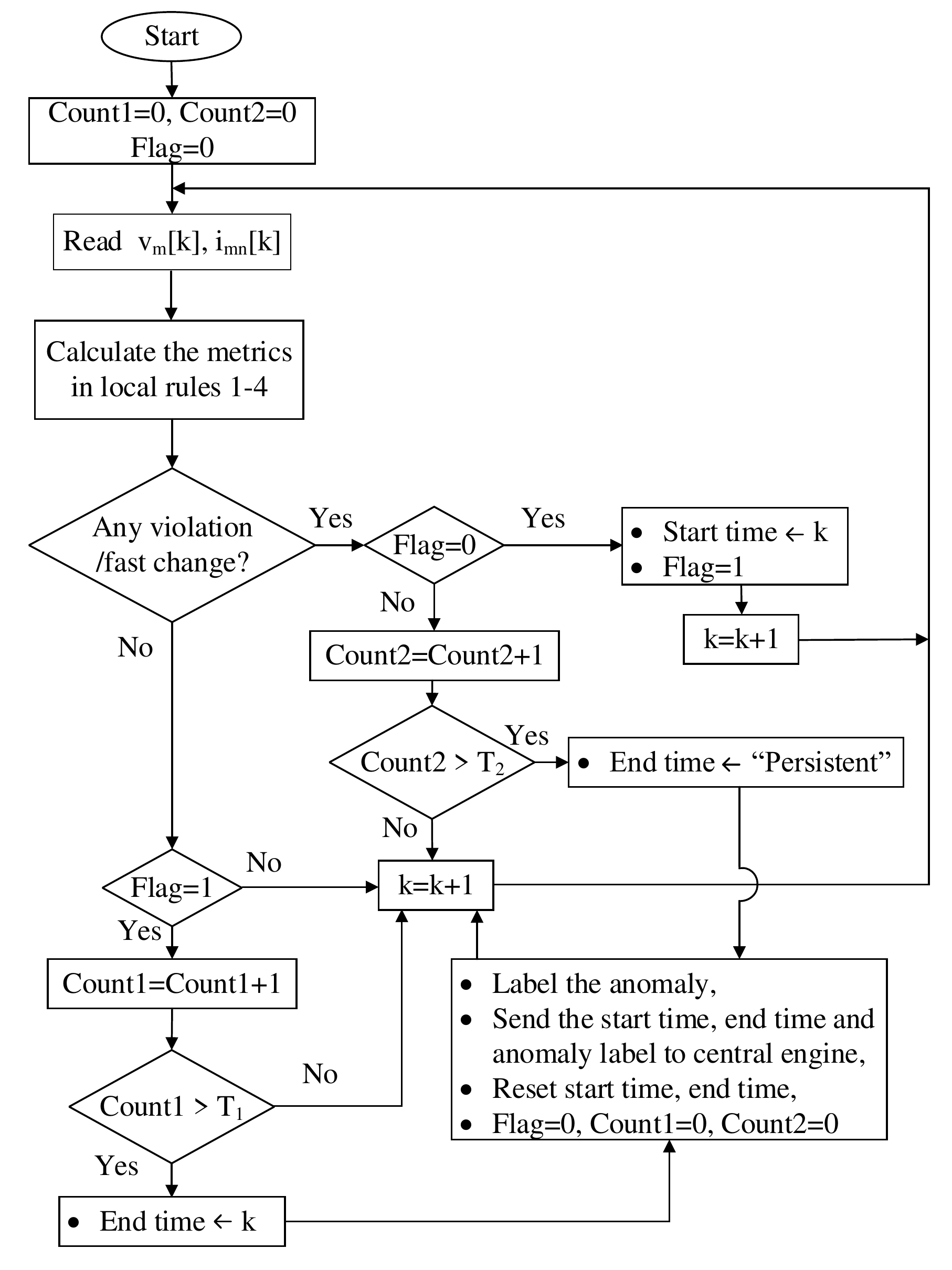}
\caption{Local Engine Analysis Flowchart Next to $\mup$ at Bus $m$.}
\label{fig.local_flowchart}
\end{figure}

\subsection{Central Rule}\label{sec.stage-i}
At higher levels of aggregation, the central engine in our case, the natural way to relate the measurements is through the grid interconnection. 
For a grid with $B$ buses, let $\mathbf{I}[k]$ denote the vector of three-phase current injection phasors with size $3B$, and $\mathbf{V}[k]$ represent the vector of three-phase voltage phasor at all the buses, which contains $3B$ elements. We define the measurement vector $\mathbf{d}[k]=(\mathbf{I}^T[k],\mathbf{V}^T[k])^T$.
During the steady-state, the following algebraic equation should hold:
\begin{align}
\mathbf{H}\mathbf{d}[k]
=\mathbf{0},~~~ 
\mathbf{H}&=\left( \begin{array}{c|c}
\mathcal{I}_{3B}  &-\bs{Y}_{3(B \times B)}
\end{array} \right)
\label{eq:grid_homogen}
\end{align} 
where $\bs{Y}$ is the admittance matrix of the grid that connects the current injection to the bus voltages. During the quasi steady-state these equations are close to be homogeneous. Since the distribution grid is unbalanced, and the lines are not transposed the set of equations that should be dealt with are three phase instead of working with positive sequence \cite{kersting2012distribution}. In this framework, we are also able to include the laterals in the admittance matrix by putting the entries corresponding to the phases that do not exist equal to zero.  

It is assumed that $\mathbf{H}$ is known, $K$ denotes the number of $\mup$s that are available and that each device has enough channels to measure the voltage and all incidental current measurements of the bus on which it is placed. Therefore, having a $\mup$ at bus $m$ means that the entries in $\mb{V}[k]$ corresponding to $\mb{v}_m[k]$  and entries in $\mb{I}[k]$ corresponding to $\sum_{n:m \sim n} \mb{i}_{mn}[k]$ are both available, where $m \sim n$ denotes that bus $m$ and $n$ are connected through a line. 
Let ${\bf T}$ denote a matrix parsing the vector $\mathbf d[k]$ into two parts corresponding to the unavailable measurements, $\mathbf{d}_u[k]$, and the available measurements, $\mathbf{d}_a[k]$:
\begin{equation}
{\mathbf T}=
\begin{pmatrix}
{\mathbf T}_u\\
{\mathbf T}_a
\end{pmatrix}
~\rightarrow~ 
{\mathbf T}\mathbf d=\begin{pmatrix}
{\mathbf d}_u\\
{\mathbf d}_a
\end{pmatrix},
~\mathbf{H}{\mathbf T}^T=
\left(\!\!
\begin{array}{c|c}
{\mathbf H}_u&{\mathbf H}_a
\end{array}\!\!
\right)
\end{equation}
where $K'=B-K$ and
\begin{align}
 \mb{T}_u \in \{0,1\}^{6(K' \times B)},~~
 \mb{T}_a \in \{0,1\}^{6(K \times B)}
\end{align}
$\mb{T}_u$ and $\mb{T}_a$ here are block diagonal matrices of size $6(K' \times B)$ and $6(K \times B)$ with entries equal to 0 or 1. The former is supposed to select the unavailable current and voltage phasors by having entries equal to 1 at corresponding locations, and the latter is supposed to pick the available phasors. Since ${\mathbf T}^T{\mathbf T}={\cal I}$, we can rewrite \eqref{eq:grid_homogen} as follows:
\begin{subequations}
\begin{align}
\label{eq:split_grid_homogen_1}
\mathbf{H}_u\mathbf{d}_u[k]&+\mathbf{H}_a\mathbf{d}_a[k]=\mathbf{0} \rightarrow \\
\mathbf{H}_a\mathbf{d}_a[k]&=-\mathbf{H}_u\mathbf{d}_u[k].
\label{eq:split_grid_homogen_2}
\end{align}
\end{subequations}
Let first assume that we have enough $\mup$s that satisfy $K > \frac{B}{2}$. Therefore, the matrix $\mb{H}_u$ would be a tall matrix and has a left null-space. Premultiplying both sides of \eqref{eq:split_grid_homogen_2} by the projector on the left null-space of $\mb{H}_u$, we have:
\begin{align}
(\mathcal{I}-\mb{H}_u \mb{H}^\dagger_u)\mathbf{H}_a\mathbf{d}_a[k]=\mb{0}
\label{eq:central_proj}
\end{align}
because $(\mathcal{I}-\mb{H}_u \mb{H}^\dagger_u)\mathbf{H}_u$=$\mb{0}$. The equality in \eqref{eq:split_grid_homogen_2} only holds during the steady-state (and for quasi steady-state with a good approximation), which means that equation \eqref{eq:central_proj} is homogeneous only in the steady-state (nearly-homogeneous in the quasi steady-state), and non-homogeneous otherwise. The following metric should, therefore, be close to zero only during normal operation and if $\mathbf{H}$
is unchanged:
\begin{align}
x[k]=\frac{||(\mathcal{I}-\mathbf{H}_u\mathbf{H}_u^\dagger)\mathbf{H}_a\mathbf{d}_a[k]||_2^2}{||\mb{d}_a[k]||_2^2}
\end{align}
However, in reality, the number of available $\mup$s, $K \ll \frac{B}{2}$. In this case, $\mb{H}_u$ is a fat matrix and, in general, is of full row rank, which in turn means that $(\mathcal{I}-\mathbf{H}_u\mathbf{H}_u^\dagger)$=$\mb{0}$, and our criterion becomes trivial. However, $\mb{H}_u\mb{H}^H_u$ has generally a high condition number, due to the weak grid connectivity. Therefore, both sides of \eqref{eq:split_grid_homogen_2} can be projected on the subspace spanned by the left singular vector, denoted by $\mb{u}_{u,s}$ corresponding to the smallest singular value of $\mb{H}_u$ expecting that $\mb{u}^H_{u,s}\mathbf{H}_u\mathbf{d}_u[k]$ to be small.
Accordingly, when the quasi steady-state is the regime of operation, i.e., when the equality in \eqref{eq:split_grid_homogen_2} holds, it is expected that the following function to be small and to vary smoothly:
\begin{align}
x[k]=\frac{|\mathbf{u}^H_{u,s}\mathbf{H}_a\mathbf{d}_a[k]|^2}{||\mathbf{d}_a[k]||^2}.
\label{eq:xopt2}
\end{align}
The exit from this behavior is then marked as an anomaly in the central engine 
and $x[k]$ is the quantity that is proposed to be tracked for fast changes for this purpose.

The flowchart for the central engine is similar to the one for local engine to some extent except that all the $\mup$ readings are required to form the central metric, and the results from the local analysis are received by the central engine, and the analysis results from this stage are not shipped anywhere, and are ready to be displayed for the operator. 
\subsection{Fast Change Detection Method}
\label{sec:change_det}
As explained above, some of the criteria defined in local and central rules require tracking fast changes in the quantities that are defined, because severe variations in $x[k]$ are signatures of an anomaly. From real data and simulations we have verified that variations in the mean value for these quantities during the quasi steady-state regime are extremely smooth. This observation prompted us to consider changes in their mean value as the common statistical trade-mark of anomalies in all of these quantities and to use the sequential two-sided Cumulative Sum (CUSUM) algorithm \cite{page1954continuous,basseville1993detection} as a heuristic. The use of this algorithm amounts to approximating the samples, for all the aforementioned quantities, as outcomes of a Gaussian non-zero mean process with independent observations. Although the observations are in fact temporally correlated, our objective (i.e. fast change detection in mean) justifies the relaxation that the random process has independent observation samples. The algorithm decides between two hypotheses $\mathcal{H}_0$: {\it no change in the mean}, or $\mathcal{H}_1$: {\it change in the mean}, at time $k$.

We expect to see multiple change points during an event. Detection of multiple change points is achieved by resetting the decision functions and cumulative sums to zero after the change is detected, and continuing the inspection of upcoming samples. The fast change anomaly is completed if no new changes are detected for a defined window of time.
\section{Optimal $\mup$ Placement}
In tandem with the anomaly detection, the criterion that was described for the central engine can be the basis to determine an optimal placement for the $\mup$s. As it was mentioned, the challenge here is that we cannot feasibly nor practically deploy $\mup$s at all nodes in the distribution system, therefore we consider a limited deployment, where $K \ll \frac{B}{2}$.  

Ideally, we want the matrix $\mb{H}_u$ to have a left null-space, i.e., the criterion in \eqref{eq:xopt2} to be zero in the steady-sate. Considering the high condition number of $\mb{H}_u \mb{H}^\dagger_u$, a reasonable approach to find the optimal configuration is to minimize the norm in \eqref{eq:xopt2} over all the possible placement configurations. 
Let first rewrite the defined metric in \eqref{eq:xopt2}:
\begin{align}
x[k]=\frac{\mathbf{d}_a^H[k]\mathbf{H}^H_a \mb{u}_{u,s}\mb{u}^H_{u,s} \mathbf{H}_a\mathbf{d}_a[k]}{||\mb{d}_a||^2}.
\end{align} 
We desire our formulation to be only topology-dependent. Therefore,  the optimal placement problem is formulated as a min-max optimization with the following structure:    
\begin{align}
\label{eq.optimization}
\begin{split}
\mathbf {\Pi}^\text{opt}=&~
\mbox{arg}\!\min_{\mathbf \Pi}~~ \lambda_{\max}(\mb{W}) 
\\
\mbox{s.t.}
~~&\left(\!\!
\begin{array}{c|c}
{\mathbf H}_u&{\mathbf H}_a
\end{array}\!\!
\right)=\mathbf H \left(\!\!
\begin{array}{c|c}
{\mathbf T}^T_u&{\mathbf T}^T_a
\end{array}\!\!\right)
\\
&\mb{T}=\mathcal{I}_2 \otimes (\mb{\Pi} \otimes \mathcal{I}_3),
\\
&\mb{W}=\mathbf{H}^H_a \mb{u}_{u,s}\mb{u}^H_{u,s} \mathbf{H}_a
,~~~[\mathbf{\Pi}]_{i,j} \in \{0,1\}\\
&
\sum_{j} [\mathbf{\Pi}]_{i,j}=1,~~~\sum_{i} [\mathbf{\Pi}]_{i,j} = 1
\end{split}
\end{align} 
Since $\lambda_{\max}(\mb{W})=\max_{\mb{d}_a[k]}\frac{\mathbf{d}_a^H[k] \mb{W} \mathbf{d}_a[k]}{||\mb{d}_a[k]||^2}$, essentially we are choosing a placement that minimizes the maximum value that our objective function can take over possible set of available measurement vectors $\mb{d}_a[k]$.

An exhaustive search is required to find the global optimum of the optimization problem in  \eqref{eq.optimization}, which is exponentially complex, and therefore does not scale well. This becomes a barrier when the size of the grid is large, i.e., in most of the real grids. Therefore, we propose to employ a \textit{``Greedy Search''} as an alternative to reduce the time complexity to \textit{polylog}, while accepting to be near-optimal. The pseudo-code of the employed greedy search is illustrated in Algorithm.~\ref{alg.greedy}.
\SetKwBlock{Init}{Initialization}{end}
\begin{algorithm}
\Init{
$K$ := Number of $\mup$s\;
$\mathcal{P}:=\emptyset$, //~Set of selected placement locations\;
$\mathcal{L}:=$~Set of candidate placement buses\;
} 
\Begin {
\For{n=1..K}{
$Cost \leftarrow \inf$\;
\For{each $l \in \mathcal{L}$}{
$\mathcal{P}:=\mathcal{P} \cup \{l\}$\;
given $\mathcal{P},\text{ calculate } \lambda_{\max}(\mb{W})$\; 
\If{$\lambda_{\max}(\mb{W})<Cost$}{
$l_{opt} \leftarrow l$\;
$Cost \leftarrow \lambda_{\max}(\mb{W})$\;} 
$\mathcal{P}:=\mathcal{P} \setminus \{l\}$\;}
$\mathcal{P}:=\mathcal{P} \cup \{l_{opt}\}$\;
$\mathcal{L}:=\mathcal{L} \setminus \{l_{opt}\}$\;}
}
\caption{Greedy Search Pseudo-Code for Optimal $\mup$ Placement.}

\label{alg.greedy}
\end{algorithm}


\section{Numerical Results}
\label{sec:res}
In this section, we first find the optimal placement for our $\mup$s based on \eqref{eq.optimization} and then test the effectiveness of the proposed anomaly detection criteria through simulated data and real data, provided from the $\mup$s that are installed in our partner utility medium voltage (12.47 kV) grid.
\subsection{Synthetic Data}
The IEEE-34 bus test case \cite{ieee34} is simulated using the time-domain simulation environment of DIgSILENT \cite{manual2009version} that deals with differential equations rather than memory-less equations. The sampling rate is selected to be equal to the sampling rate of the Analog-to-Digital Converter (ADC) in a real $\mup$, which is 512 $\times$ 60 Hz =30720 samples per sec.. We then processed these time-domain data through our phasor estimation algorithm, that emulates a two-cycle, P-class algorithm based on the IEEE C37.118.1 \cite{c37} producing phasor samples at a rate of 120 Hz. The single-line diagram of the test case is shown in Fig.~\ref{fig.ieee34}. This case includes single-phase laterals, voltage regulators, and untransposed lines, which all are modeled exactly in our admittance matrix.    
\begin{figure}[ht]
\centering 
\includegraphics[trim = 2mm 2mm 2mm 2mm,width=0.5\textwidth]{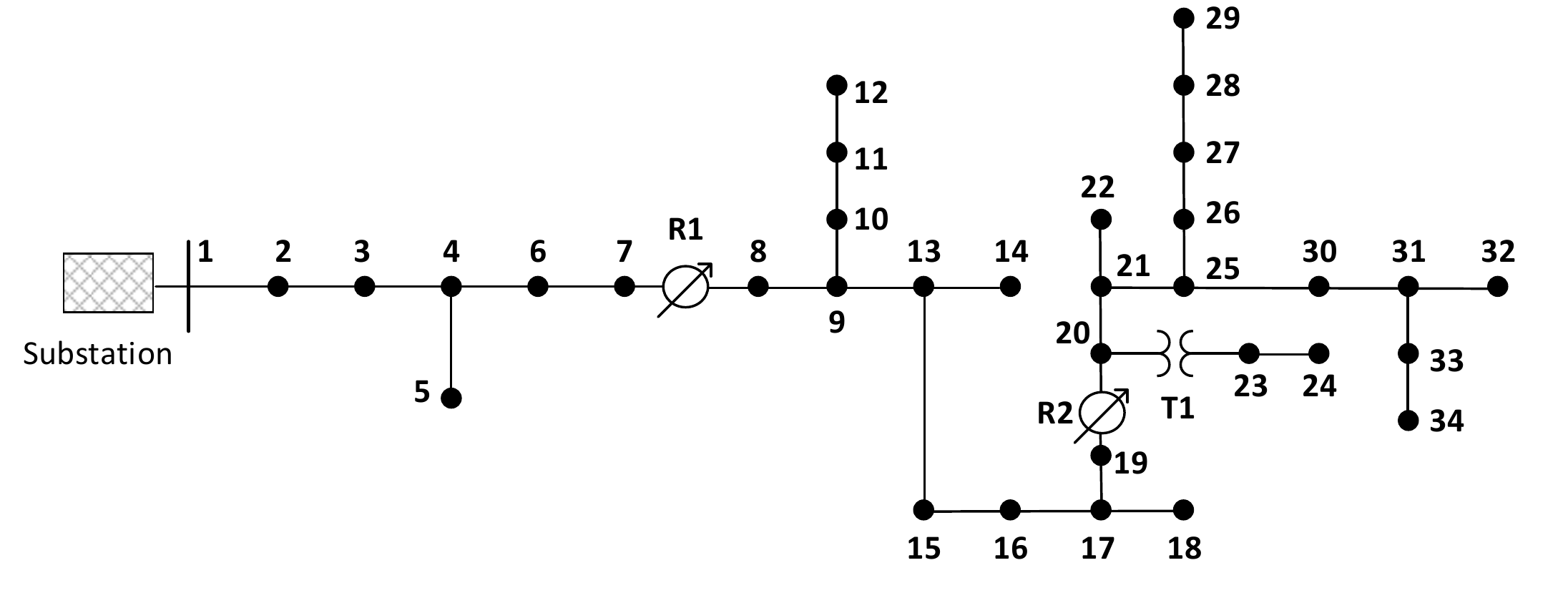}
\caption{IEEE 34-Bus Test Feeder Single-Line Diagram}
\label{fig.ieee34}
\end{figure}

Table.~\ref{tab.opt} compares the objective value for a random placement, \textit{``Greedy Search''} and \textit{``Exhaustive Search''}, and the time complexity of each solver, assuming that $K=3$ $\mup$s are available. 
\begin{table}[htbp]
\caption{IEEE-34 Case Optimal $\mup$ Placement Result for $K=3$}
\label{tab.opt}
\begin{center}
\begin{tabular}{|c|c|c|c|}
\hline
& \bf{Random} & \bf{Greedy} & \bf{Exhaustive}\\
\hline \hline
\bf{Optimum Cost} & 1.7085 &0.51477&0.51477 \\ \hline
\bf{Buses with $\mup$s} &\{1,3,9\} &\{7,19,31\}&\{9,19,31\}\\ \hline 
\bf{Run Time} & --&2.84 s&290.266 s\\ \hline 
\end{tabular}
\end{center}
\end{table}     
The objective value of the \textit{``Greedy Search''} and the \textit{``Exhaustive Search''}, and the set of the selected buses are close to each other, while the run time of the \textit{``Greedy Search''} is 102.206 times faster. Hence, the \textit{``Greedy Search''} can be a very good choice to solve our optimal placement problem.         

As expected, the placement rule tries to scatter the available $\mup$s all over the grid, in order to achieve the maximum possible coverage. 

In order to investigate how the time complexity grows, and also analyze the results of the placement criterion, the 123 standard test case in \cite{ieee34} was used, considering 20 $\mup$s available (i.e., $K=20$). Without loss of generality, the 123 test feeder was reduced to 70 buses to only include three-phase lines, and roll up all the laterals (single-phase and two-phase lines). The reasoning behind this reduction is that visibility on the main feeder is more important for us than the the visibility on the laterals, considering the limited number of $\mup$s. Using a  machine with 60 Intel(R) Xeon(R) CPU E7-4870 v2 @ 2.30GHz cores, it took 11.56 Sec for the algorithm to place 20 $\mup$s over 70 buses. Fig.~\ref{fig.placement123} illustrates the location of the optimally-placed $\mup$s in the reduced grid. It can be observed from the figure that the $\mup$s are scattered over the grid to achieve the maximum sensitivity with respect to different locations of anomalies. 
\begin{figure}[ht]
\centering 
\includegraphics[width=0.5\textwidth]{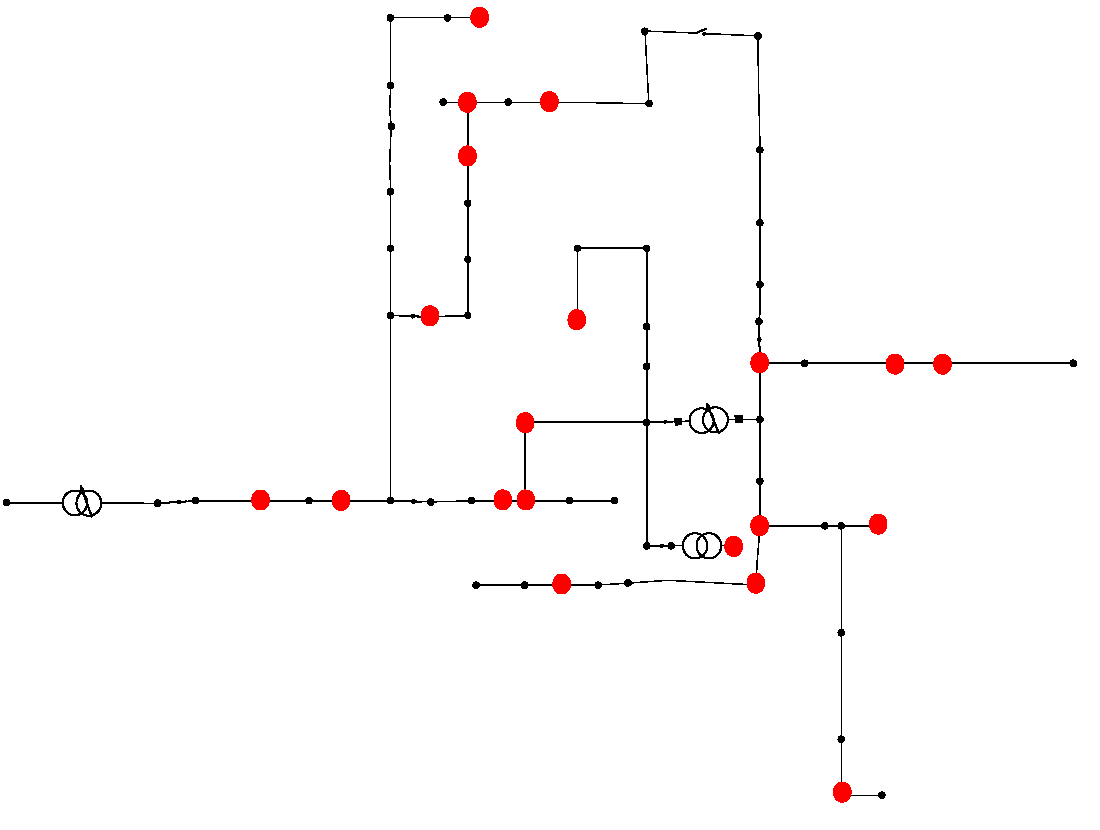}
\caption{ Optimally-Placed $\mup$s on Reduced IEEE 123 Test Feeder.}
\label{fig.placement123}
\end{figure}      

\subsubsection{Single-Line to Ground Fault-IEEE 34 Bus} In this simulation, the rules are tested for detecting anomaly with respect to a
Single-Line to Ground Fault (SLGF), which is a very common type of short-circuit fault in the distribution grid. A SLGF was introduced on ``Phase a'' of line (25,26), which then caused the fuse placed on the phase a of this line near bus 25 to melt down. Our three $\mup$s are placed on buses, 7, 19, and 31 based on the ``greedy search'' result for the optimal placement criterion. 

The results of the \textit{``voltage magnitude change''} rule is shown in Fig.~\ref{fig.sim_vol_mag} for $\mup$s 7, and 19 for instance. Simply, the rule inspects the data for large deviations, label them accordingly, and find the start time and the end time of the event, marked with blue and red stars.
  
\begin{figure}[ht]
\centering 
\includegraphics[width=0.5\textwidth,height=4 cm]{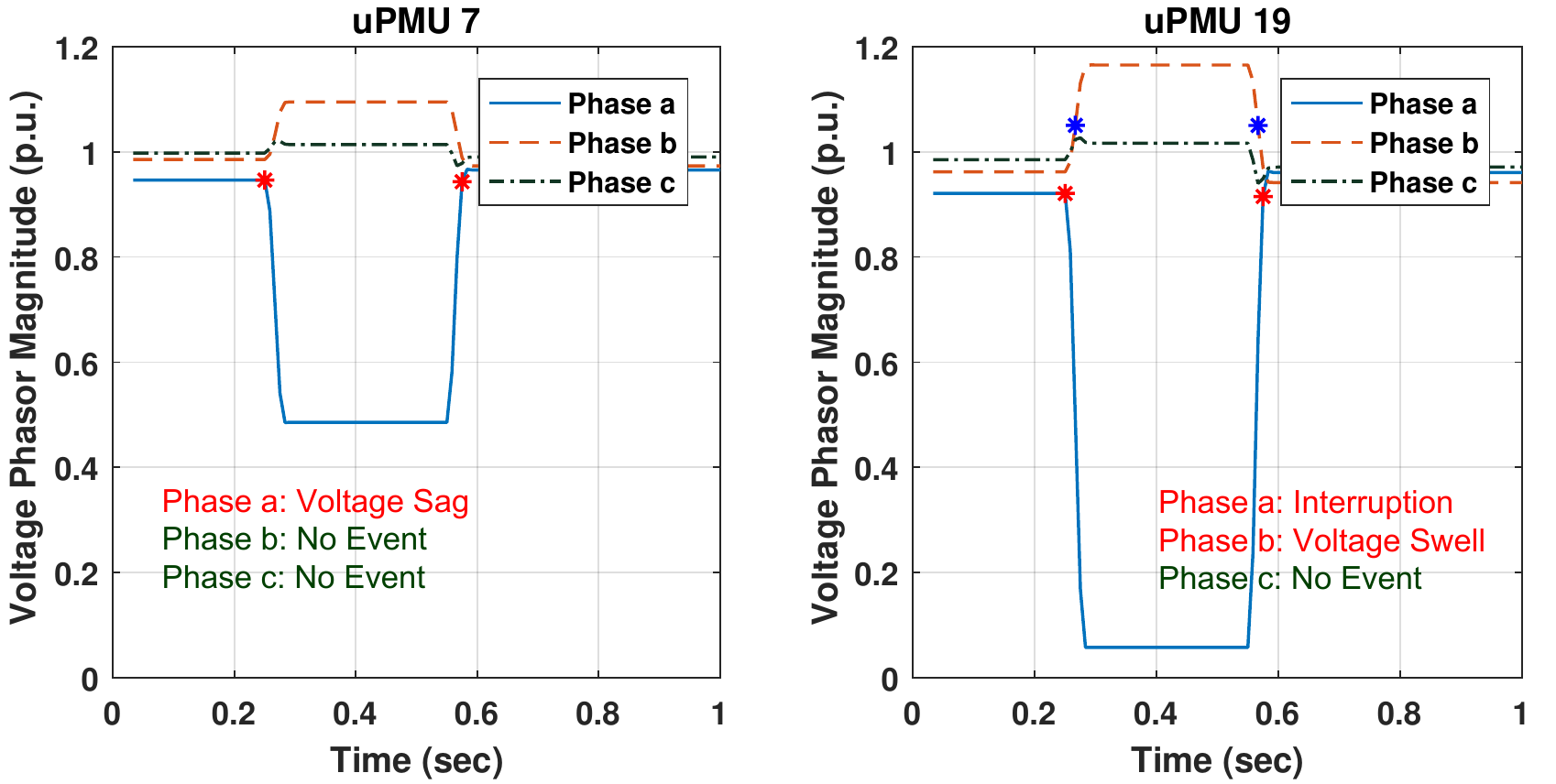}
\caption{ Voltage Magnitude Local Rule Result for SLGF.}
\label{fig.sim_vol_mag}
\end{figure}
Fig.~\ref{fig.sim_SS_local} illustrates the metric value in \eqref{eq:xopt2} for the specified lines that is inspected with $M=12$ to check whether the grid is in the quasi steady-state or not, where the voltage and current data are first converted to per-unit system assuming $S_b=1$ MVA. The start time of the detected changes are also marked, setting the CUSUM detector parameters fixed in all the three local engines corresponding to each $\mup$. As it can be observed, there are two periods in which the grid manifests its dynamic: the first one corresponds to the occurrence of the fault and the second matches with the fuse meltdown. In addition, based on the severity of the transient that each $\mup$ measures, the number of detected changes via CUSUM varies. In this case, considering the location of the $\mup$s and the location and type of the fault, the most severe change appears in the metric corresponding to measurements from line (19,20), while the changes in metric for line (31,32) is very small. Therefore, based on the defined parameters on the detector, CUSUM finds quite a large number of change points in the former, while it is not set to be sensitive to the changes in the order that appears in the latter. In fact, if the detector is set to be too sensitive, it can increase ``\textit{false alarms}'' in the system. 
Also note that due to the two-cycle calculation of the phasor, and use of $M$ samples to calculate the correlation matrix, the event appears and disappears with a systematic delay.      
\begin{figure}[ht]
\centering 
\includegraphics[width=0.5\textwidth]{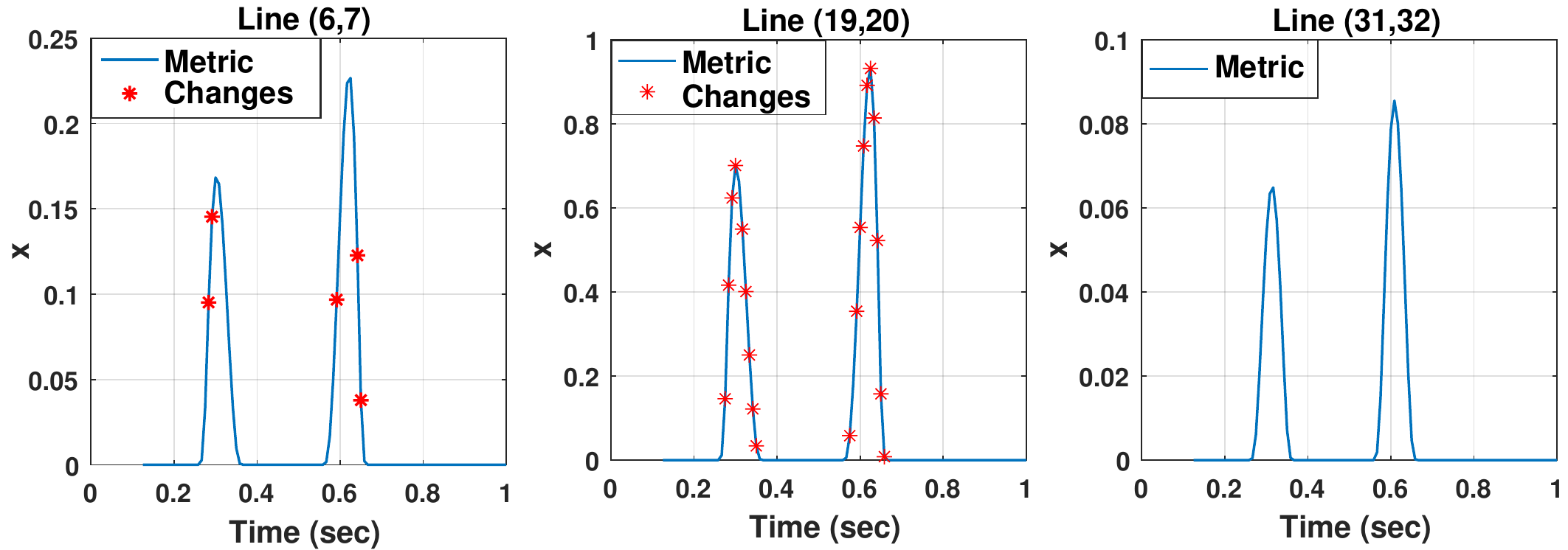}
\caption{Quasi Steady-State Validity Checking for SLGF.}
\label{fig.sim_SS_local}
\end{figure} 
The other local rules also capture the anomaly, though with different severity and behavior (we just show some of the results here for lack of space). In fact, many of the local rules may detect the same event, though some rules are more informative than others depending on the cause. Each triggered rule reports a start and an end time for the event. Storing these time-tags for eventful segments of data helps understanding their relationship.

The metric defined for our central engine is also illustrated in Fig.~\ref{fig.sim_SS_cent}. The delay in appearing and disappearing of the event in here is solely due to the two-cycle phasor calculation.
\begin{figure}[ht]
\centering 
\includegraphics[trim = 2mm 2mm 2mm 2mm,width=0.75\linewidth,height=4 cm]{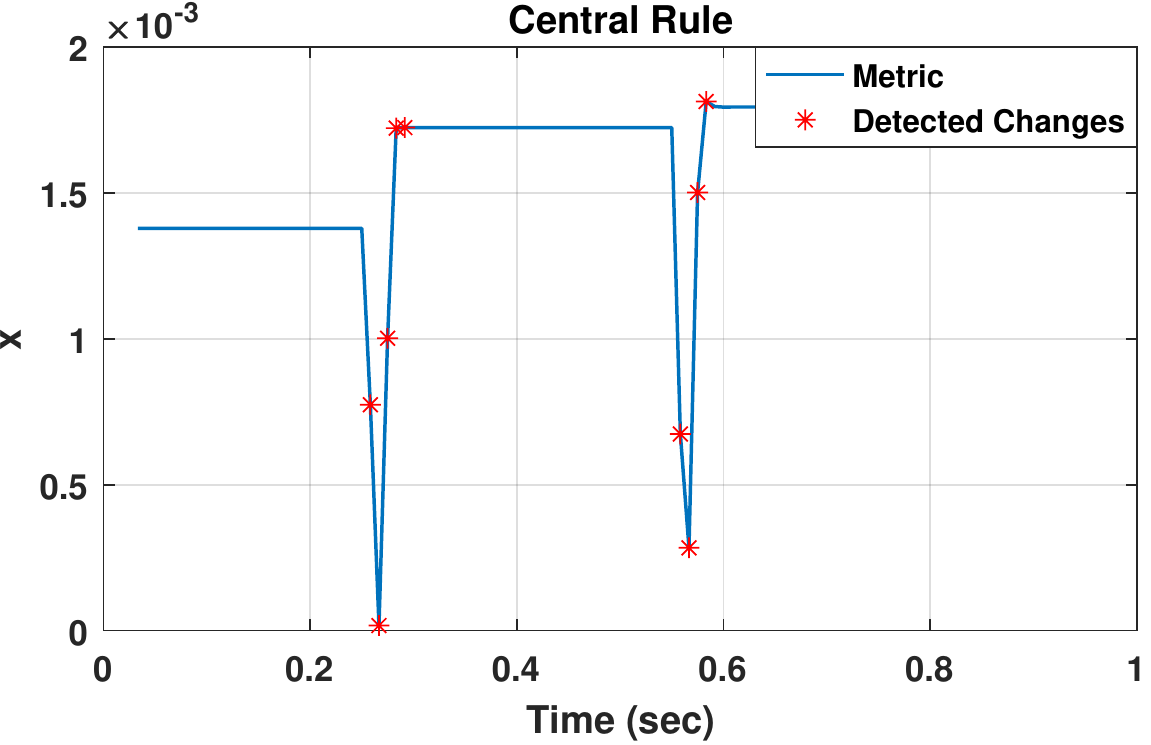}
\caption{Central Rule Inspection for SLGF.}
\label{fig.sim_SS_cent}
\end{figure} 

It should be noted that the value of the metric in \eqref{eq:xopt2} during the quasi steady-state is highly dependent on the number of $\mup$s and the topology of the grid that together determine the structure of $\mb{H}_u$ and $\mb{H}_a$. 
However, the best placement of the $\mup$s for a certain topology and number of $\mup$s that makes the metric to be as close as possible to zero, returns a certain objective value that can be used as the baseline to determine what is a normal value for $x[k]$ and what would be an anomaly.   
\subsubsection{Discussion--Compromised Data} 
We illustrate here the resilience of the architecture to data injection attacks. 
For this purpose, three data attack scenarios are investigated happening concurrently with the SLGF event discussed previously. We consider the case of the attacker manipulating the data of $\mup$ 7 in the first case, and $\mup$ 19 in the second case, and finally $\mup$ 7 and 19 at the same time in the third case on their way to the central engine. In all cases the data injected are a {\it replay} of the last available data set before the anomaly starts. Setting the change detector parameters fixed for all the three cases, Fig.~\ref{fig.sim_central_compromised} shows, for each case, the central rule and the start time of the detected changes. As can be seen, since $\mup$ 19 is playing an important role for this event, having $\mup$ 7 compromised will not affect our central rule significantly (case-1). However, when the $\mup$ 19 is compromised, the number of detected changes reduce significantly (case-2), and when both $\mup$ 7 and 19 are compromised and the only healthy data is coming from $\mup$ 31, the detector does not pick any fast changes based on the set parameters. This also reveals the importance of tuning the detector thresholds to have a certain ``false alarm,'' while maximizing the ``detection'' probability.   
\begin{figure}[ht]
\centering 
\includegraphics[trim = 0mm 0mm 0mm 0mm,width=0.5\textwidth]{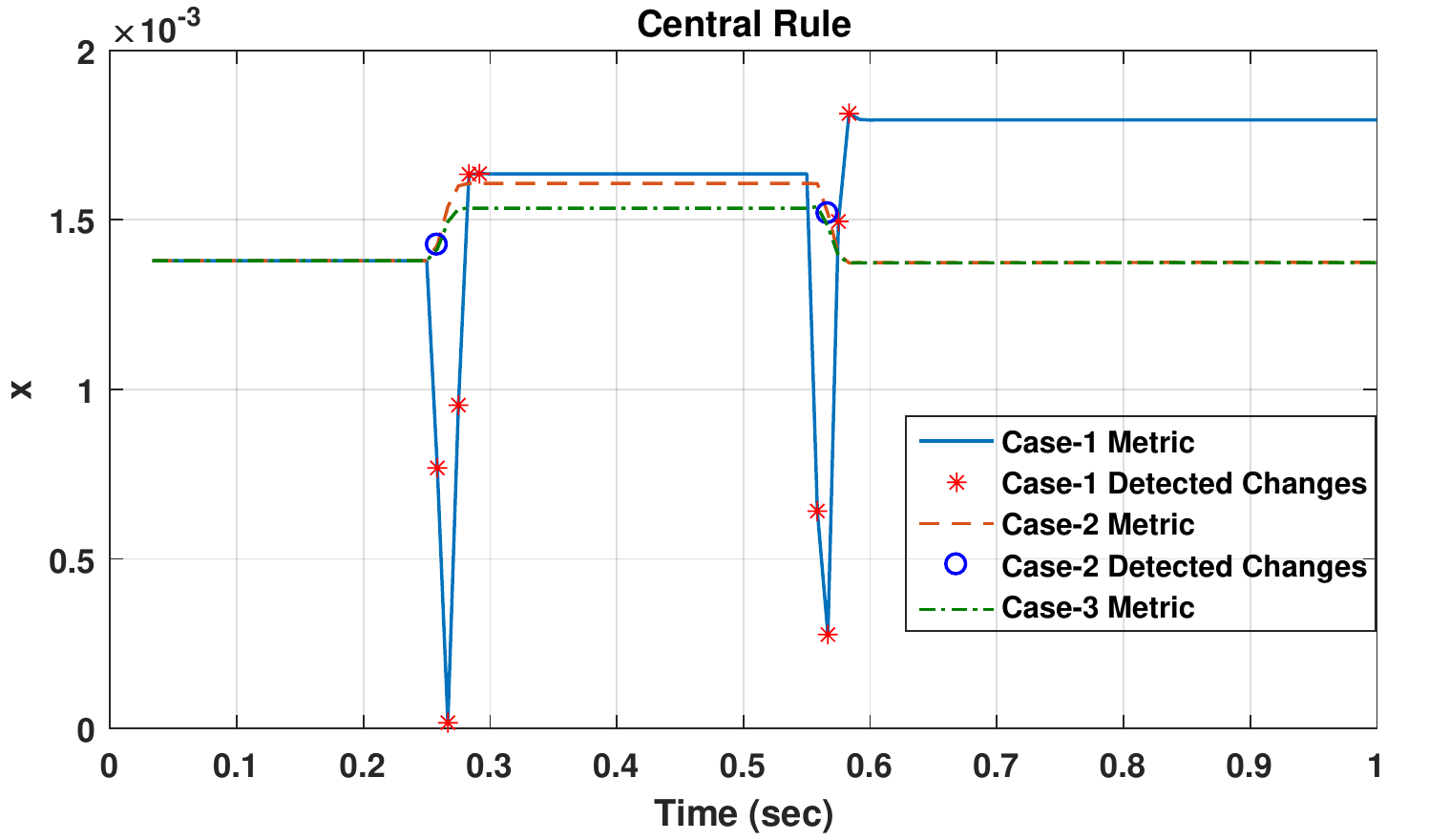}
\caption{Central Rule for SLGF with Manipulated $\mup$ Data.}
\label{fig.sim_central_compromised}
\end{figure}          
We wish to remark that in all the three cases, the local analytics that directly draw data from $\mup$s will still flag the alarm, so buffering locally at the site of the event these data can be an important way of helping understand what communications were compromised in an ex-post analysis.
\subsubsection{Discussion--Optimal vs Non-Optimal Placement} 
As it was mentioned, the placement criterion tries to scatter the available $\mup$s over the grid to achieve the maximum coverage, and therefore make the central rule more sensitive to anomalies. 

In order to compare the performance of an optimal versus non-optimal placement, a load loss event is created on bus 24 of IEEE-34 test case at $t=0.4$s. The $\mup$s are placed based on the random placement and greedy search result given in Table.~\ref{tab.opt}, corresponding to non-optimal and optimal placement, accordingly. Fig.~\ref{fig.greedyVSrandom} shows the central metric for these two cases. Since the relative change is what matters to our detector, the metrics corresponding to two placement are brought on the same scale. As it can be observed, since the $\mup$s in the random placement are concentrated at a certain area, events like this would not be very pronounceable in the central metric, which could possibly lead to a ``false-negative''.
\begin{figure}[ht]
\centering 
\includegraphics[trim = 2mm 2mm 2mm 2mm,width=0.75\linewidth]{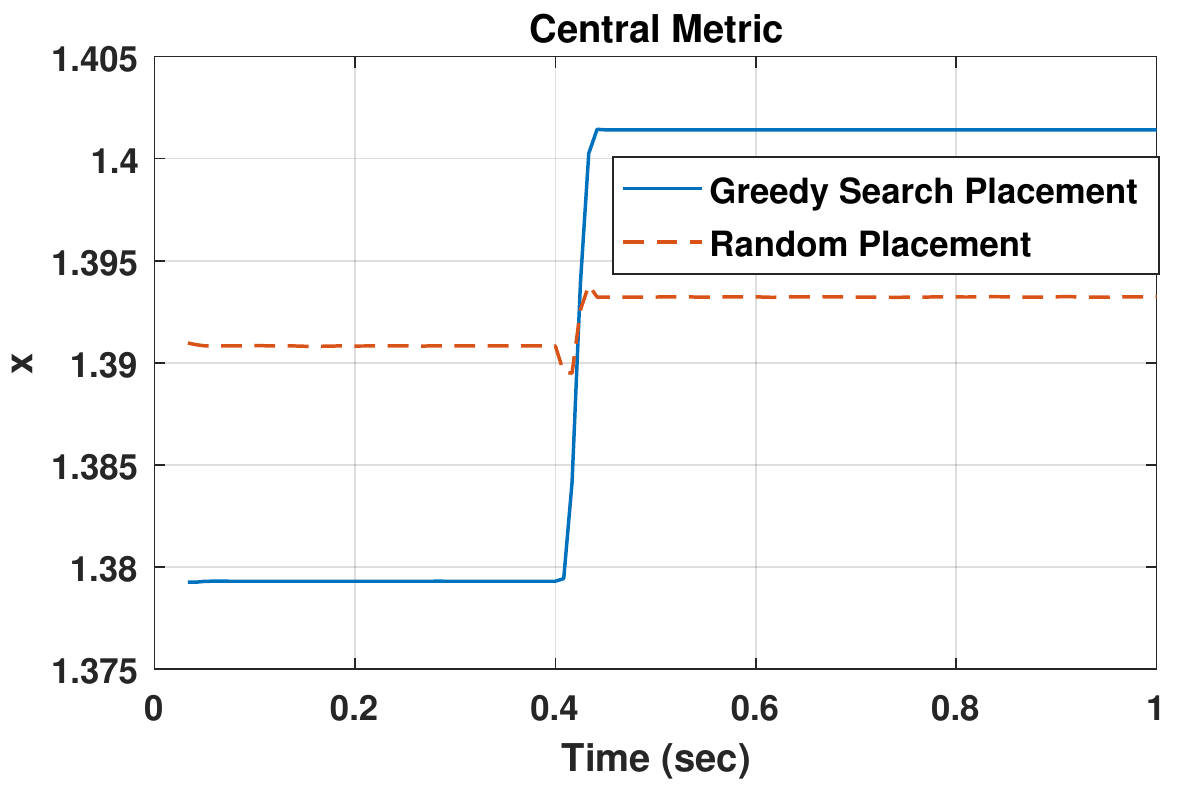}
\caption{Central Metric Change of a Load Loss Event; Optimal versus Non-optimal Placement.}
\label{fig.greedyVSrandom}
\end{figure}   
\subsection{Real Data}
Fig.~\ref{fig.rpu} shows the abstract one-line diagram of the partner utility grid and the location of the installed $\mup$s. The installed $\mup$s sample the voltage and current with a rate of $512 \times 60$~Hz, and output the estimated phasors at 120 Hz rate. These devices achieve an accuracy of 0.001 deg resolution for phasor angle, 0.0002\% for phasor magnitude, and 0.01\% for Total Vector of Error (TVE) \cite{upmu_site}. The two feeders here are connected through the subtransmission grid. 
\begin{figure}[ht]
\centering 
\includegraphics[trim = 0mm 0mm 0mm 0mm,width=0.5\textwidth]{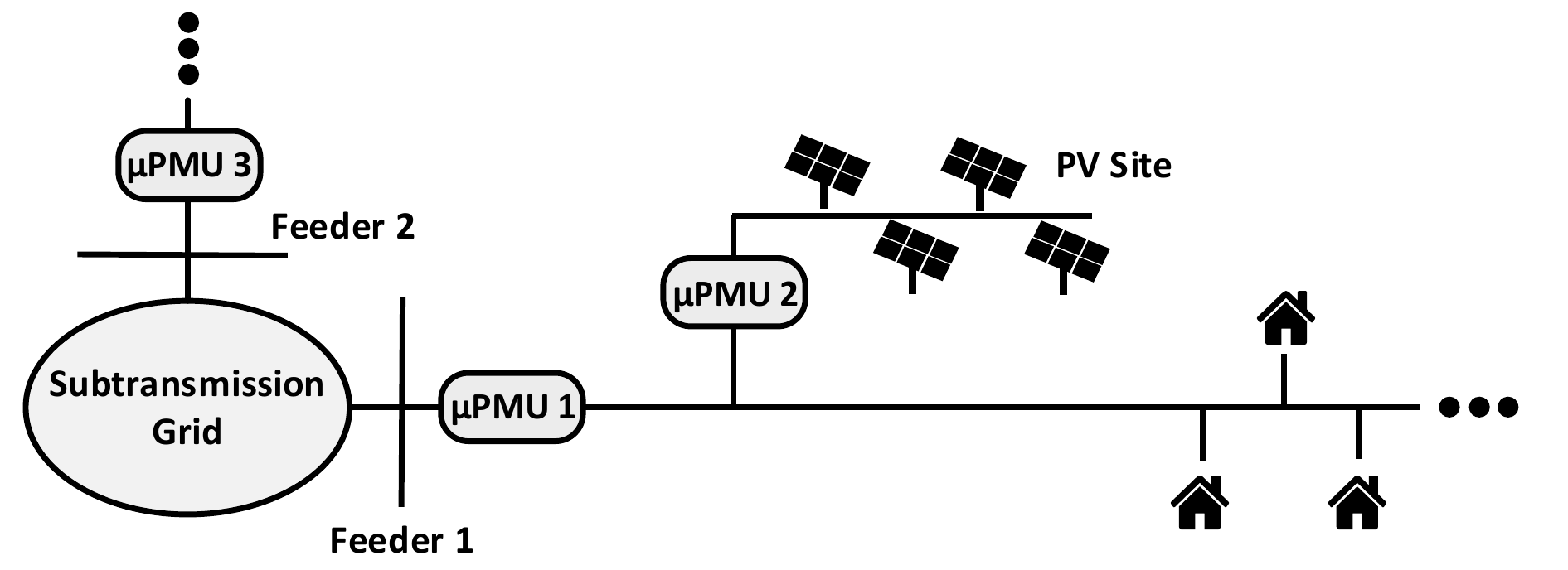}
\caption{Location of Installed $\mup$s in Our Partner Utility Grid.}
\label{fig.rpu}
\end{figure}   

Fig.~\ref{fig.voltage_sag_05_jan} shows the voltage magnitude change rule inspected on the data from these $\mup$s over a certain time.  
\begin{figure}[ht]
\centering 
\includegraphics[width=0.5\textwidth,height=6 cm]{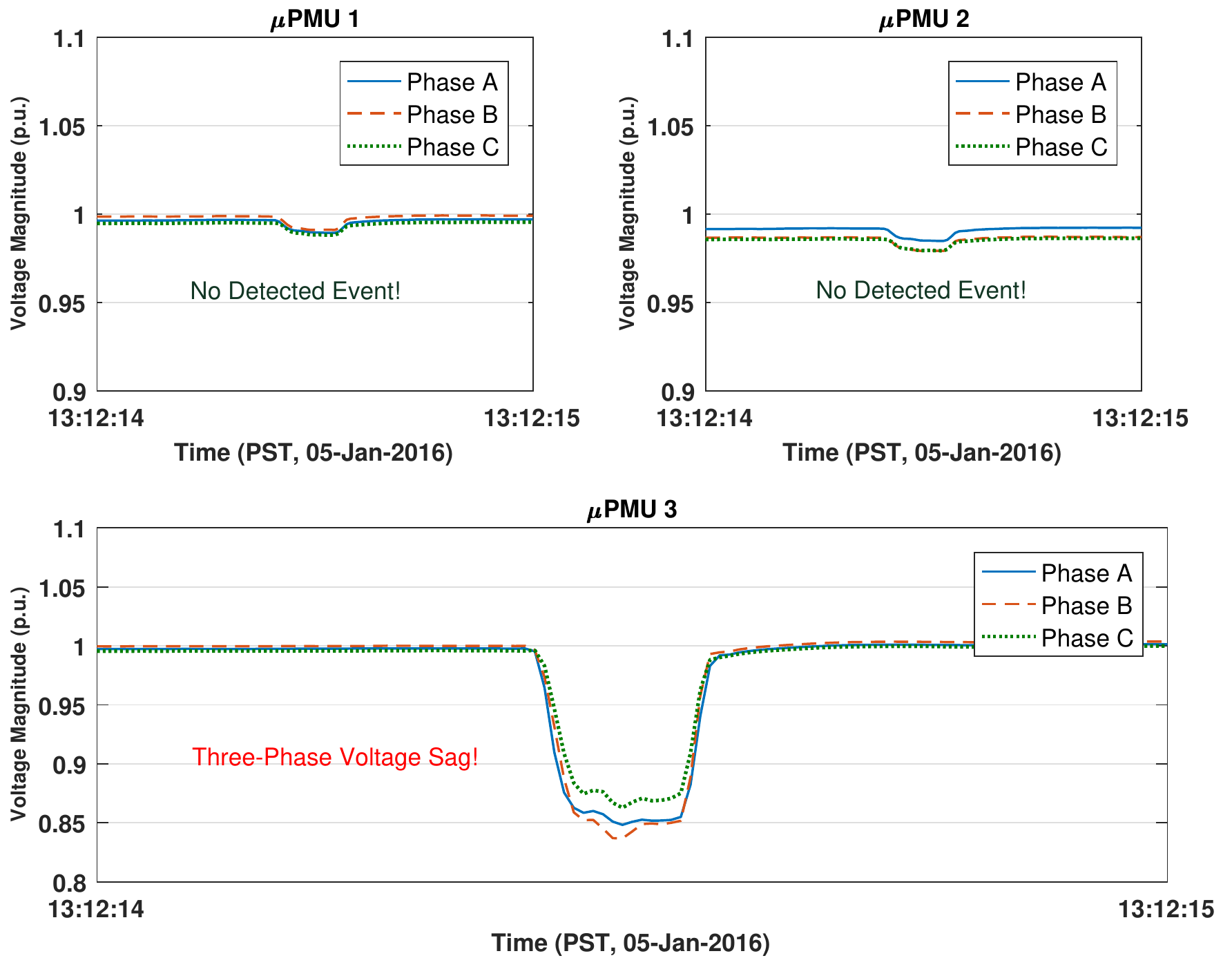}
\caption{Voltage Magnitude Change Rule for Real Data.}
\label{fig.voltage_sag_05_jan}
\end{figure}
The results of the fast change inspection on the current magnitude of phase a and instantaneous local frequency at the substation bus of feeder 2 are also illustrated in Fig.~\ref{fig.freq_current_05_jan}(a) and Fig.~\ref{fig.freq_current_05_jan}(b), respectively. 
\begin{figure}[ht] 
\centering
    \includegraphics[width=0.5\textwidth,height=4 cm]{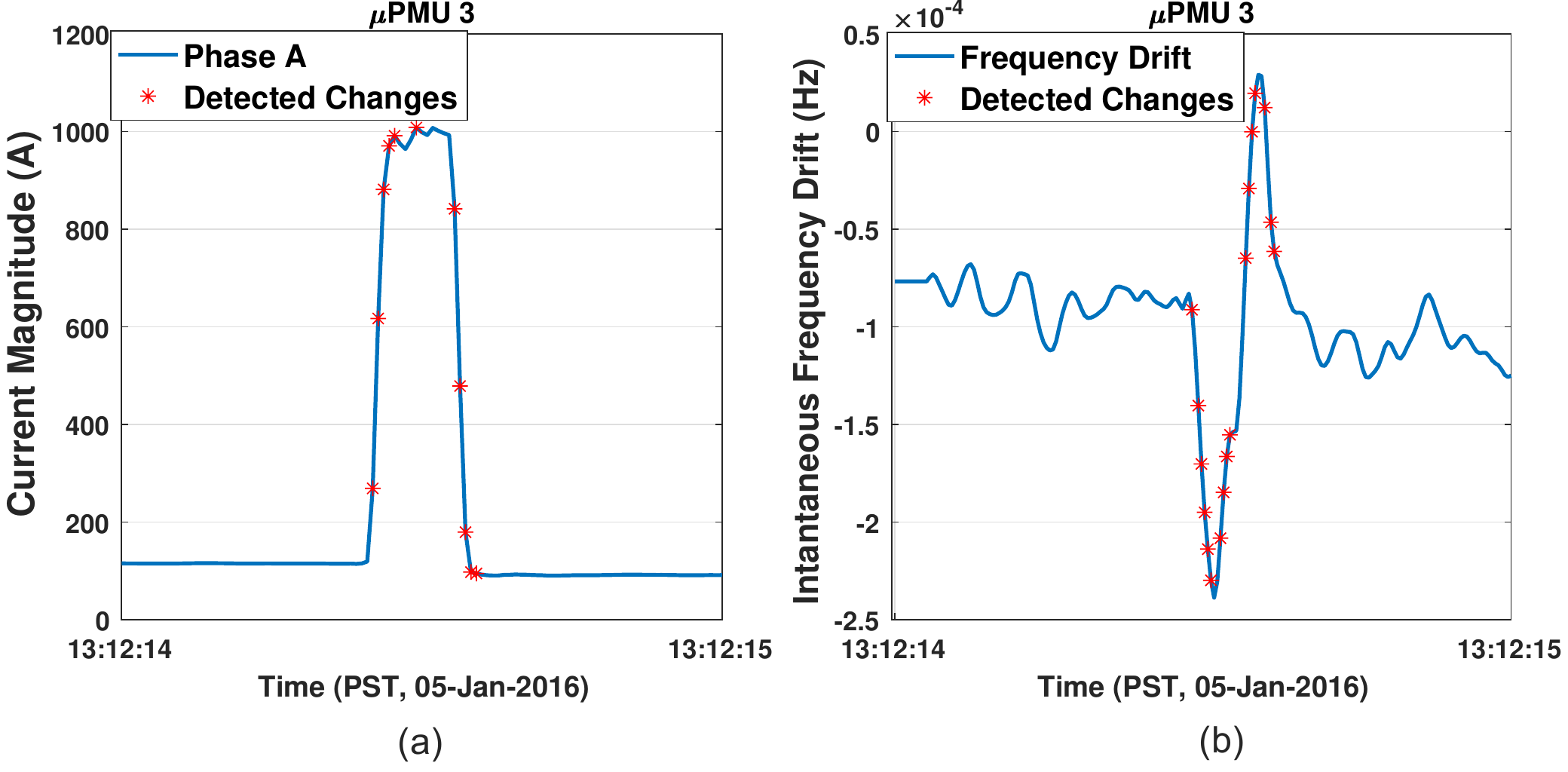}
    \caption{Fast Change Tracking of Current Phasor Magnitude and Bus Instantaneous Frequency Drift for Real Data.}
    \label{fig.freq_current_05_jan}
\end{figure} 
Observing the results of analysis, the DSO can deduce that the cause of the event is most probably located on feeder 2. Also, from the pre- and post-anomaly value of the current magnitude, the DSO can conclude that some of the loads on this feeder tripped due to the voltage sag. 
We just showed the results from some of the rules due to the space limit but other rules can also flag the existence of anomaly on feeder 2.

All the metrics introduced and tested above are designed considering the specifications of a distribution grid. It should be noted that not all the proposed methods in the literature for transmission grid are applicable in the distribution side. For example, the phase angle difference in the distribution grid is known to be much smaller than that in the transmission level. Therefore, as opposed to transmission grid that this metric would work well for event detection \cite{allen2014pmu}, it might not be a proper metric to look at in the distribution grid, since the signal to noise ratio might be small. The example next illustrates how the voltage angle difference metric could have failed if it was used as a ``local rule''. 
Using the data from two $\mup$s installed at two ends of a line at a second utility grid (not the one in Fig.~\ref{fig.rpu}), the voltage magnitude captured by the $\mup$ at one end of the line is shown in Fig.~\ref{fig.ang_diff}(a) and the voltage phasor angle difference between the two $\mup$s at two ends of the line is shown in Fig.~\ref{fig.ang_diff}(b). As it can be observed, the important event in this period corresponds to the two voltage sags. However, the angle difference shows a significant number of spikes without clearly marking these two events with the same significance. Passing this metric to our detector, we would pick too many fast changes that do not represent any specific event of interest, which therefore means an increase in the number of ``false positives''. All streams of interest were examined during this period, including both the active and reactive power, and all anomalies detected were in agreement with those visible in Fig.~\ref{fig.ang_diff}(a)

\begin{figure}[ht] 
\centering
\includegraphics[trim = 2mm 2mm 2mm 2mm,width=0.5\textwidth]{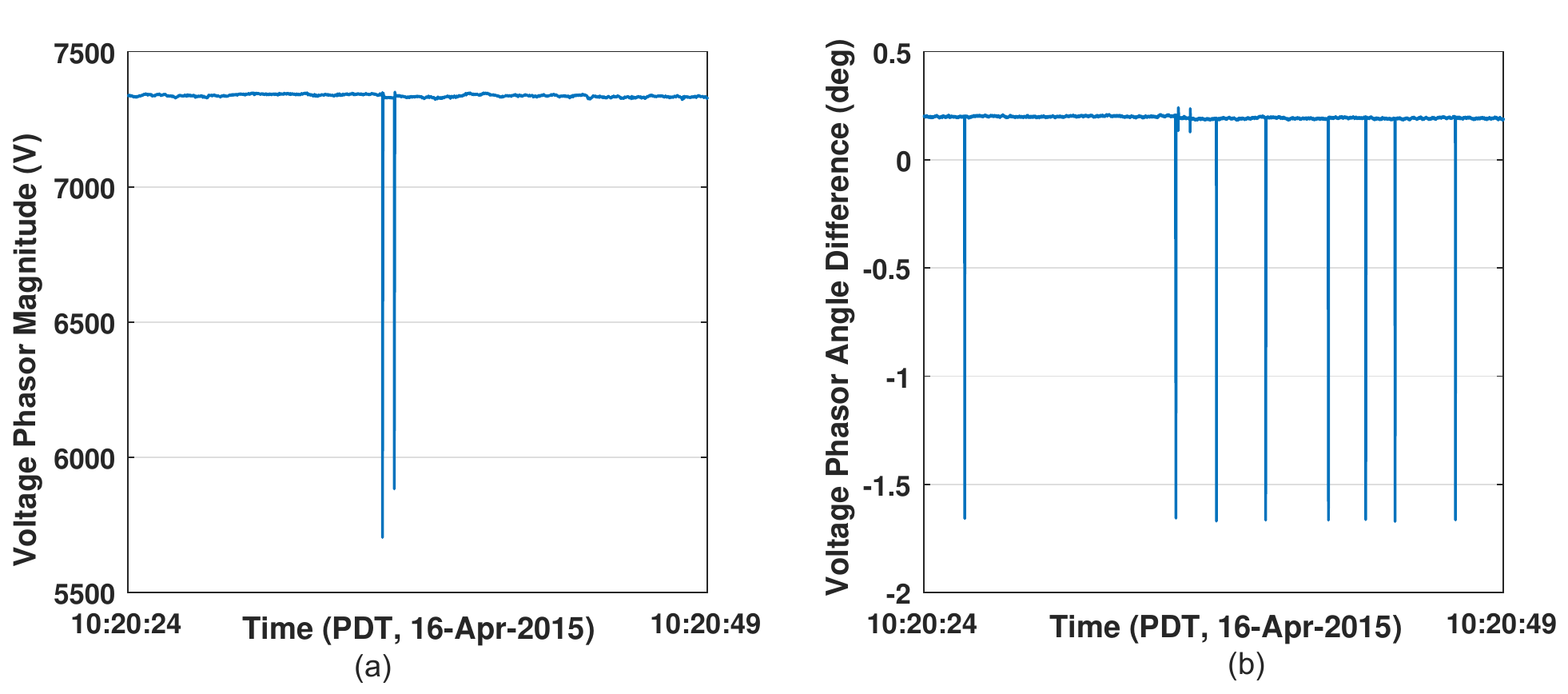}\caption{Voltage Phasor Angle Difference Between Two $\mup$s at Distribution Grid.}
\label{fig.ang_diff}
\end{figure}  
\subsubsection{Case Study-Robustness Against Volatility}
The increasing presence of renewable resources, like wind and solar, in the distribution grid  may raise some concerns about the robustness of the proposed rules to ``false alarms'' due to their inherent volatility. However, our method is based on an adaptive estimation of the data mean using an exponential window and, as shown in the following numerical example, this makes the algorithm capable to only pick the fast changes over a very short period of time, while being insensitive to changes in the time order of normal renewable resources fluctuations. To show it numerically, we refer to two real events captured by $\mup$ 2 in the grid in Fig.\ref{fig.rpu}. Fig.~\ref{fig.active_power_PV}(a) shows a dramatic change of about -18.25\% over 9 seconds in the PV site active power injection due to the cloud effects, and Fig.~\ref{fig.active_power_PV}(b) shows a step change in the active power after the PV site went  out of service. Setting the detection parameters to be the same for both events, it can be seen that the detector flags no event for the first case and finds multiple of fast changes in the second. This is ideal since the events of our interest in the this context are of the second type, which should be discriminated from the normal quasi steady-state behavior in the first case.    
\begin{figure}[ht] 
\centering
    \includegraphics[width=0.5\textwidth,height=4 cm]{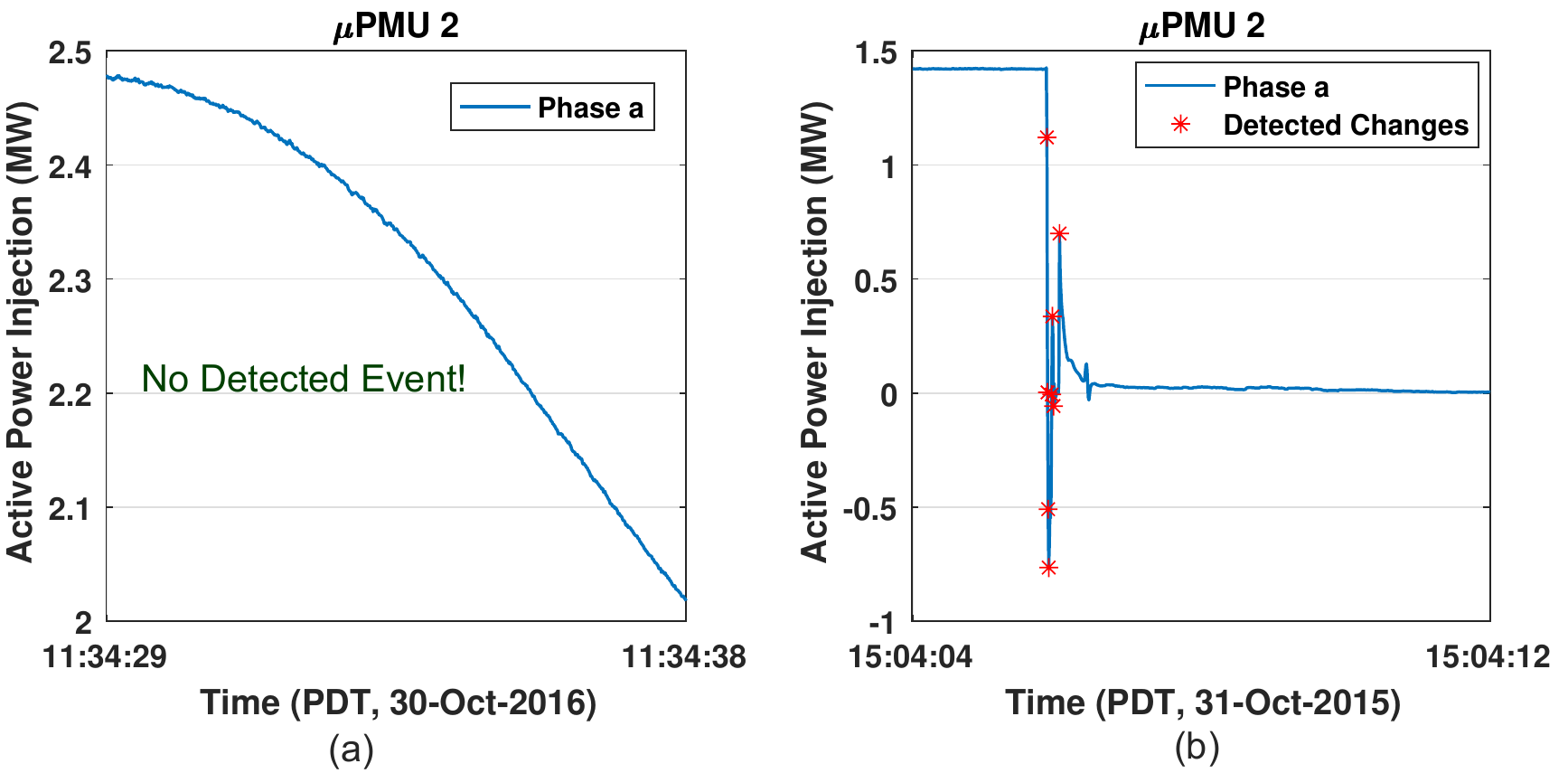}
	\caption{Fast Change Tracking of Injected Active Power in the PV Site.}
    \label{fig.active_power_PV}
\end{figure}

Considering the existence of different types of loads in the distribution grid, the detector should also be able to differentiate between possible normal fluctuations in the load profile and those that are caused by a rare event. To show the performance of CUSUM for this case, real data of a $\mup$ installed behind a building with a non-linear load is used. As it can be seen in Fig.~\ref{fig.current_bank514}, the fast change tracker of the current magnitude does not pick the fluctuations due to the non-linear load behavior as an event and only flags a segment of data that corresponds to a voltage sag in the grid as an eventful segment.
\begin{figure}[ht] 
\centering
    \includegraphics[width=0.5\textwidth,height=4 cm]{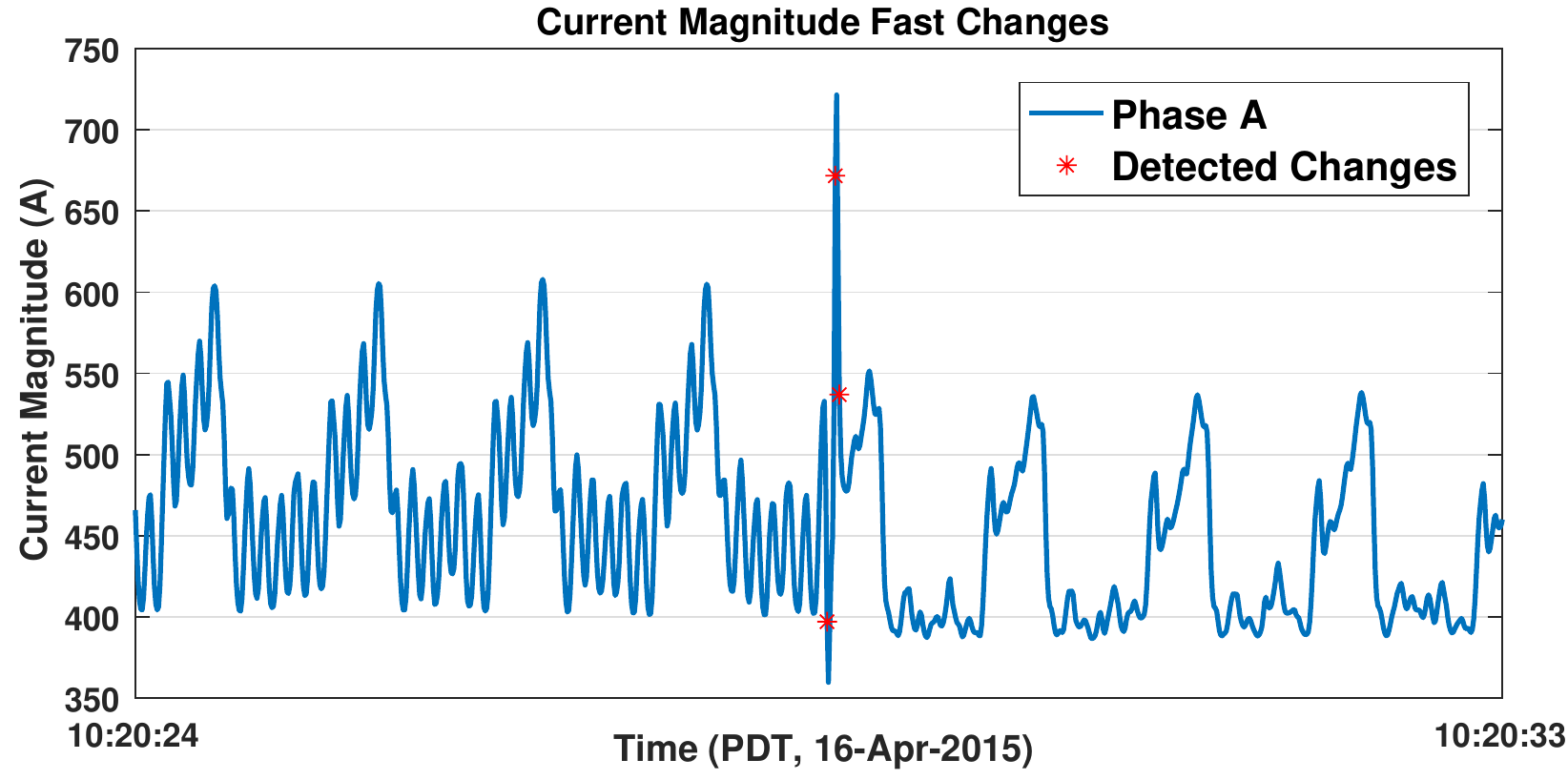}
	\caption{Fast Change Tracking of Current Magnitude for a Non-Linear Load.}
    \label{fig.current_bank514}
\end{figure}

\section{Conclusion and Future Direction}
In this paper, a hierarchical \textit{``anomaly detection architecture''} has been described that integrates high resolution, synchronous $\mup$ data to attain high accuracy in the detection of the physical effects of an event. 
A set of rules has been formulated to process the $\mup$s' data at the local and central level of our hierarchical architecture. Exiting from the quasi steady-state regime is a common signature of many anomalies, and this insight is exploited to design a key part of the proposed metrics. To ensure scalability of our architecture, the rules at the local stage are designed to be agnostic about the grid interconnection and physical location of the sensor. Depending on the type of event, the effects on different metrics has different severity and, therefore, inspecting different metrics at the local level reduces the miss-detection rate. The analysis in the higher stages is designed to bind the readings from different $\mup$s using the knowledge about the grid topology and sensor locations in search for eventful segments.    

Using real data and simulations, we have shown how processing real-time physical measurements from $\mup$s enables us to detect effectively anomalies in the system and inform the operator about the grid status. 
Having a limited number of $\mup$s, an optimal placement criterion is also designed with respect to the central rule that aims to scatter $\mup$s over the grid to achieve maximum coverage and therefore to increase the sensitivity of the rule to different event source locations.

Having an anomaly identified, our future efforts will focus on identifying and localizing the source of the anomaly using the underlying physical model when possible or, in general, at least confine the search region to a smaller set of components/lines. This include the identification of any change in the admittance matrix, due to loss of a line, reconfiguration of the grid, and etc.    
\bibliography{bib}

\begin{thebibliography}{10}
\providecommand{\url}[1]{#1}
\csname url@samestyle\endcsname
\providecommand{\newblock}{\relax}
\providecommand{\bibinfo}[2]{#2}
\providecommand{\BIBentrySTDinterwordspacing}{\spaceskip=0pt\relax}
\providecommand{\BIBentryALTinterwordstretchfactor}{4}
\providecommand{\BIBentryALTinterwordspacing}{\spaceskip=\fontdimen2\font plus
\BIBentryALTinterwordstretchfactor\fontdimen3\font minus
  \fontdimen4\font\relax}
\providecommand{\BIBforeignlanguage}[2]{{%
\expandafter\ifx\csname l@#1\endcsname\relax
\typeout{** WARNING: IEEEtran.bst: No hyphenation pattern has been}%
\typeout{** loaded for the language `#1'. Using the pattern for}%
\typeout{** the default language instead.}%
\else
\language=\csname l@#1\endcsname
\fi
#2}}
\providecommand{\BIBdecl}{\relax}
\BIBdecl

\bibitem{jamei2016automated}
M.~Jamei, A.~Scaglione, C.~Roberts, E.~Stewart, S.~Peisert, C.~McParland, and
  A.~McEachern, ``Automated anomaly detection in distribution grids using
  $\mu$pmu measurements,'' \emph{arXiv preprint arXiv:1610.01107}, 2016.

\bibitem{von2014micro}
A.~von Meier, D.~Culler, A.~McEachern, and R.~Arghandeh, ``Micro-synchrophasors
  for distribution systems,'' in \emph{Proc. IEEE PES Innovative Smart Grid
  Technologies Conference (ISGT)}, 2014, pp. 1--5.

\bibitem{scoping_study}
J.~H. Eto, E.~M. Stewart, T.~Smith, M.~Buckner, H.~Kirkham, F.~Tuffner, and
  D.~Schoenwald, \emph{Scoping Study on Research and Priorities for
  Distribution-System Phasor Measurement Units}.\hskip 1em plus 0.5em minus
  0.4em\relax Lawrence Berkeley National Laboratory, 2015.

\bibitem{kezunovic2013role}
M.~Kezunovic, L.~Xie, and S.~Grijalva, ``The role of big data in improving
  power system operation and protection,'' in \emph{2013 IREP Symposium Bulk
  Power System Dynamics and Control - IX Optimization, Security and Control of
  the Emerging Power Grid}, Aug 2013, pp. 1--9.

\bibitem{phadke2008wide}
A.~Phadke and R.~M. de~Moraes, ``The wide world of wide-area measurement,''
  \emph{Power and Energy Magazine, IEEE}, vol.~6, no.~5, pp. 52--65, 2008.

\bibitem{terzija2011wide}
V.~Terzija, G.~Valverde, D.~Cai, P.~Regulski, V.~Madani, J.~Fitch, S.~Skok,
  M.~M. Begovic, and A.~Phadke, ``Wide-area monitoring, protection, and control
  of future electric power networks,'' \emph{Proceedings of the IEEE}, vol.~99,
  no.~1, pp. 80--93, 2011.

\bibitem{pan2015developing}
S.~Pan, T.~Morris, and U.~Adhikari, ``Developing a hybrid intrusion detection
  system using data mining for power systems,'' \emph{Smart Grid, IEEE
  Transactions on}, vol.~6, no.~6, pp. 3104--3113, 2015.

\bibitem{chen2013dimensionality}
Y.~Chen, L.~Xie, and P.~R. Kumar, ``Dimensionality reduction and early event
  detection using online synchrophasor data,'' in \emph{Power and Energy
  Society General Meeting (PES), 2013 IEEE}.\hskip 1em plus 0.5em minus
  0.4em\relax IEEE, 2013, pp. 1--5.

\bibitem{xie2014dimensionality}
L.~Xie, Y.~Chen, and P.~Kumar, ``Dimensionality reduction of synchrophasor data
  for early event detection: Linearized analysis,'' \emph{IEEE Transactions on
  Power Systems}, vol.~29, no.~6, pp. 2784--2794, 2014.

\bibitem{valenzuela2013real}
J.~Valenzuela, J.~Wang, and N.~Bissinger, ``Real-time intrusion detection in
  power system operations,'' \emph{Power Systems, IEEE Transactions on},
  vol.~28, no.~2, pp. 1052--1062, 2013.

\bibitem{ge2015power}
Y.~Ge, A.~J. Flueck, D.-K. Kim, J.-B. Ahn, J.-D. Lee, and D.-Y. Kwon, ``Power
  system real-time event detection and associated data archival reduction based
  on synchrophasors,'' \emph{IEEE Transactions on Smart Grid}, vol.~6, no.~4,
  pp. 2088--2097, 2015.

\bibitem{allen2014pmu}
A.~Allen, M.~Singh, E.~Muljadi, and S.~Santoso, \emph{PMU Data Event Detection:
  A User Guide for Power Engineers}.\hskip 1em plus 0.5em minus 0.4em\relax
  National Renewable Energy Laboratory, 2014.

\bibitem{biswalsupervisory}
M.~Biswal, S.~M. Brahma, and H.~Cao, ``Supervisory protection and automated
  event diagnosis using pmu data,'' \emph{IEEE Transactions on Power Delivery},
  vol.~31, no.~4, pp. 1855--1863, Aug 2016.

\bibitem{jamei2016micro}
M.~Jamei, E.~Stewart, S.~Peisert, A.~Scaglione, C.~McParland, C.~Roberts, and
  A.~McEachern, ``Micro synchrophasor-based intrusion detection in automated
  distribution systems: Toward critical infrastructure security,'' \emph{IEEE
  Internet Computing}, vol.~20, no.~5, pp. 18--27, 2016.

\bibitem{brahmareal}
S.~Brahma, R.~Kavasseri, H.~Cao, N.~R. Chaudhuri, T.~Alexopoulos, and Y.~Cui,
  ``Real time identification of dynamic events in power systems using pmu data,
  and potential applications - models, promises, and challenges,'' \emph{IEEE
  Transactions on Power Delivery}, vol.~PP, no.~99, pp. 1--1, 2016.

\bibitem{phadke2008synchronized}
A.~G. Phadke and J.~S. Thorp, \emph{Synchronized phasor measurements and their
  applications}.\hskip 1em plus 0.5em minus 0.4em\relax Springer Science \&
  Business Media, 2008.

\bibitem{5154067}
``{IEEE Recommended Practice for Monitoring Electric Power Quality},''
  \emph{IEEE Std 1159-2009 (Revision of IEEE Std 1159-1995)}, pp. c1--81, June
  2009.

\bibitem{xia2012widely}
Y.~Xia and D.~P. Mandic, ``Widely linear adaptive frequency estimation of
  unbalanced three-phase power systems,'' \emph{IEEE Transactions on
  Instrumentation and Measurement}, vol.~61, no.~1, pp. 74--83, 2012.

\bibitem{kersting2012distribution}
W.~H. Kersting, \emph{Distribution system modeling and analysis}.\hskip 1em
  plus 0.5em minus 0.4em\relax CRC press, 2012.

\bibitem{page1954continuous}
E.~Page, ``Continuous inspection schemes,'' \emph{Biometrika}, pp. 100--115,
  1954.

\bibitem{basseville1993detection}
M.~Basseville, I.~V. Nikiforov \emph{et~al.}, \emph{Detection of abrupt
  changes: theory and application}.\hskip 1em plus 0.5em minus 0.4em\relax
  Prentice Hall Englewood Cliffs, 1993, vol. 104.

\bibitem{ieee34}
http://ewh.ieee.org/soc/pes/dsacom/testfeeders/index.html.

\bibitem{manual2009version}
D.~P.~F. Manual and D.~PowerFactory, ``Version 14.0,'' \emph{DIgSILENT GmbH,
  Gomaringen, Germany}, 2009.

\bibitem{c37}
``{IEEE Standard for Synchrophasor Measurements for Power Systems},''
  \emph{IEEE Std C37.118.1-2011 (Revision of IEEE Std C37.118-2005)}, pp.
  1--61, Dec 2011.

\bibitem{upmu_site}
``Micro-phasor measurement unit brochure,''
  http://www.powersensorsltd.com/microPMU-brochure.php, accessed: 2016-01-15.

\end{thebibliography}
\vspace{-0.1cm}
\bibliographystyle{IEEEtran}
\end{document}